\documentclass[reqno]{amsart}
\usepackage{comment,amssymb,latexsym,upref,enumerate,fouridx}
\usepackage{mathrsfs,color}
\usepackage{centernot}

\usepackage[colorlinks,linkcolor=blue,citecolor=blue]{hyperref}
\usepackage{doi}
\allowdisplaybreaks

\numberwithin{equation}{section}

\theoremstyle{plain}
\newtheorem{thm}{Theorem}[section]
\newtheorem{lem}[thm]{Lemma}
\newtheorem{prop}[thm]{Proposition}
\newtheorem{cor}[thm]{Corollary}

\theoremstyle{definition}
\newtheorem{Def}[thm]{Definition}

\theoremstyle{remark}
\newtheorem{rem}[thm]{Remark}

\newcommand{\rmd}{\mathrm{d}}

\input Tdef.def

\newcounter{mnotecount}[section]
\let\oldmarginpar\marginpar
\renewcommand\marginpar[1]{\-\oldmarginpar[\raggedleft\footnotesize #1]%
	{\raggedright\footnotesize #1}}

\begin{document}

\title[Stability of AVTD Behavior within the Polarized $\Tbb^2$-symmetric vacuum spacetimes]{Stability of AVTD Behavior within the Polarized $\Tbb^2$-symmetric vacuum spacetimes}

\author[E. Ames]{Ellery Ames}
\address{Dept. of Mathematics \\ 
Humboldt State University \\
1 Harpst St.
Arcata, CA 95521\\USA }
\email{Ellery.Ames@humboldt.edu}

\author[F. Beyer]{Florian Beyer}
\address{Dept. of Mathematics and Statistics\\
730 Cumberland St\\
University of Otago, Dunedin 9016\\ New Zealand}
\email{fbeyer@maths.otago.ac.nz }

\author[J. Isenberg]{James Isenberg}
\thanks{J. Isenberg is partially supported by NSF grant PHY-1707427.}
\address{Dept. of Mathematics and Institute for Fundamental Science \\
University of Oregon \\
Eugene, OR 97403 USA}
\email{isenberg@uoregon.edu}

\author[T.A. Oliynyk]{Todd A. Oliynyk}
\address{School of Mathematical Sciences\\
9 Rainforest Walk\\
Monash University, VIC 3800\\ Australia}
\email{todd.oliynyk@monash.edu}

\begin{abstract}
  We prove stability of the family of Kasner solutions within the class of polarized $\Tbb^2$-symmetric solutions of the vacuum Einstein equations in the contracting time direction with respect to an areal time foliation. All Kasner solutions for which the asymptotic velocity parameter $K$ satisfies $|K-1|>2$ are non-linearly stable, and all sufficiently small perturbations exhibit asymptotically velocity term dominated (AVTD) behavior and blow-up of the Kretschmann scalar. 
\end{abstract}

\maketitle

\section{Introduction}
\label{s.introduction}

While the general conjectured \emph{BKL behavior} \cite{belinskii1970,lifshitz1963} near the singularity in cosmological solutions of Einstein's equations is very difficult to investigate mathematically, the more special \emph{asymptotically velocity term dominated (AVTD) behavior} has been proven to hold for subfamilies of such solutions. 
Unlike BKL behavior, which involves a congruence of time-like observers each experiencing ``mixmaster" type behavior as they approach the singularity \cite{Weaver:2001,Andersson:2005}, in spacetime solutions with AVTD behavior, the time-like observers each experience ``spatially pointwise'' Kasner type behavior; see Definition~\ref{def.avtd} below for the precise definition of AVTD behavior.

Two major approaches for verifying the AVTD property have been employed in the mathematical general relativity literature.
The first, an approach based on the \emph{Fuchsian method}, aims to show the existence of families of solutions of the Einstein-matter or vacuum Einstein equations with AVTD behavior by prescribing asymptotics near the singularity and solving the equations in the direction away from the singularity. 
While this approach yields infinite-dimensional families of solutions parameterized by free ``asymptotic data'' functions, it does not address whether the family of solutions constructed is open (not to mention dense) in the set of all solutions under consideration, for example, within the given symmetry class, or with certain matter fields. 
To obtain this information, one typically follows a second approach based on the Cauchy problem of hyperbolic PDEs.
In this setting, the AVTD property is tightly connected with the nonlinear stability of the Kasner solutions \cite{kasner1921}.

AVTD behavior has been proven to occur generically in the class of vacuum Gowdy solutions \cite{Isenberg:1990,CIM1990,ringstrom2009a}.
Infinite dimensional subfamilies of solutions with AVTD behavior are known in more general symmetry-defined classes of spacetimes \cite{kichenassamy1998,isenberg1999,stahl2002,isenberg2002,damour2002,choquet-bruhat2004,choquet-bruhat2006,Clausen2007,ames2013a,beyer2017}, in spacetimes without assumed symmetry coupled to a stiff fluid \cite{andersson2001,heinzle2012}, and in vacuum spacetimes without assumed symmetry \cite{ChruscielKlinger:2015,Fournodavlos:2020}.
The above results are based on Fuchsian methods. 
The second approach, based on evolving a Cauchy problem towards the singularity, has recently been used to prove the stability of AVTD behavior.
Specifically, in the Einstein-scalar field and Einstein-perfect fluid spacetimes in $3+1$ dimensions it has been shown that AVTD behavior is stable for solutions in an open neighborhood of Friedman-LeMaitre-Robertson-Walker (FLRW) spacetimes \cite{rodnianski2014}, while in the Einstein-vacuum system with dimension $D+1$ for $D\ge 38$ certain moderately anisotropic Kasner solutions are nonlinearly stable \cite{Rodnianski2018HighD}.
Recent work \cite{fournodavlos2020b} removes the restriction in \cite{rodnianski2014} of being within a neighborhood of FLRW, lowers $D$ to $D \ge 10$ in the context of \cite{Rodnianski2018HighD}, and establishes the nonlinear stability of Kasner solutions (and AVTD behavior) within the class of polarized $U(1)$-symmetric vacuum solutions. 
We also mention recent work \cite{ringstrom2017,Lott:2020a,Lott:2020b,Ringstrom:2021a,Ringstrom:2021b} to establish geometric conditions under which one obtains detailed information regarding the geometry of the singularity. 

In this article, using results obtained by two of the authors with a collaborator \cite{BOOS:2020}, we show that a subfamily of the Kasner solutions (for most choices of the Kasner parameter $K$, as specified below) are nonlinearly stable within the class of the $3+1$-Einstein vacuum solutions with polarized $\Tbb^2$-symmetry, and that each such solution in an open neighborhood of this subfamily exhibits AVTD behavior. 
The polarization condition that we assume in this work is crucial. 
Numerical studies of $\Tbb^2$-symmetric solutions of the vacuum Einstein equations without the polarization condition (and with non-vanishing twist) strongly indicate that such solutions generically exhibit non-AVTD BKL behavior in a neighborhood of the singularity \cite{Weaver:2001,Andersson:2005}.
As noted above, and as this work indicates, with the polarization condition imposed, $\Tbb^2$-symmetric vacuum solutions are expected to generically show AVTD behavior near the singularity. 
We note that recent work strongly supports the conjecture that, even if the polarization condition is not imposed, $\Tbb^2$-symmetric solutions of the Einstein equations coupled to a scalar field or a stiff fluid generically exhibit AVTD behavior near the singularity \cite{rodnianski2018,rodnianski2014,fournodavlos2020b}.

Our work in this paper focuses on the behavior of polarized $\Tbb^2$-symmetric vacuum spacetimes in the contracting direction near the singularity. 
The behavior of these same spacetimes in the expanding direction has also been studied in recent work. 
While the expanding direction results obtained to date \cite{Ringstrom:2015,LeFloch:2016,Berger:2019} are less definitive, and are likely to be less predictive of the behavior of solutions with no symmetry, they do show that certain subfamilies of the $\Tbb^2$-symmetric vacuum solutions have distinctive attractors for the spacetime averaged behavior in the expanding direction.

We provide a definition of the $\Tbb^2$-symmetric (cosmological) spacetimes in Section~\ref{s.T2spacetimes}, and based on the results in \cite{BERGER1997}, which show that such spacetimes always globally admit \emph{areal coordinates}, we specify in that section the canonical form \eqref{T2metric} of the metric for $\Tbb^2$-symmetric spacetimes in areal coordinates. 
A key feature of the areal coordinate representation of these spacetimes is that the time coordinate $t$ locates the singularity at $t=0$, and as $t$ increases, the spacetime expands. 
Thus the analysis of the behavior of $\Tbb^2$-symmetric spacetimes near the singularity focuses on these solutions for positive $t$ very close to $0$. 
Also in Section~\ref{s.T2spacetimes}, we define the polarized and the Kasner subfamilies of the $\Tbb^2$-symmetric spacetimes, and we write out the vacuum Einstein equations for these spacetimes in terms of the areal coordinate metric components.

In Section~\ref{s.AVTD_spacetimes}, we define what it means for a spacetime to have AVTD  behavior in a neighborhood of the spacetime singularity. 
As noted above, AVTD behavior is a specialized version of BKL behavior. 
Early discussions of spacetimes with this property appear in \cite{Eardley:1972}. 
The formal definition of this behavior appears first in \cite{Isenberg:1990}.

We state the main result of this paper, Theorem~\ref{thm:mainresult}, in Section~\ref{s.main_result}. 
Then in Section~\ref{s.main_proof}, we proceed to carry out the details of the proof of this theorem. 
As noted above, our results in this paper depend crucially upon work presented in \cite{BOOS:2020}. 
In that work, a theorem is proven which provides sufficient conditions for a hyperbolic PDE system in Fuchsian form which guarantee that solutions that are evolved from initial data satisfying certain smallness conditions at $T_0 > 0$ must extend all the way to the singularity at $t=0$.
In the Appendix, we state this theorem, Theorem~\ref{symthm}, in a form that is most useful for the analysis carried out in this work. 
Thus the bulk of the work carried out in proving Theorem~\ref{thm:mainresult} in Section~\ref{s.main_proof} consists of first showing that the Einstein vacuum equations for the polarized $\Tbb^2$-symmetric spacetimes, when cast into first order Fuchsian form, satisfy the hypotheses of Theorem~\ref{symthm}., and then showing that we can impose conditions on the initial data so that the conclusions of Theorem~\ref{symthm} hold. 
These conditions correspond to initial data sets which are small perturbations of initial data for Kasner solutions for the Kasner parameter $K$ (see Section~\ref{s.kasner_spacetimes}) satisfying $K< -1$ or $K>3$. 

Following this application of Theorem~\ref{symthm} in Proposition~\ref{prop:main_existence}, the completion of the proof of Theorem~\ref{thm:mainresult} in Section~\ref{s.main_proof} involves improving the estimates for the asymptotics (Proposition~\ref{prop:hoexp}), and the verification that these improved asymptotic estimates for the solutions imply the stability of Kasner solutions within the $\Tbb^2$-symmetric class (Section~\ref{s.stability_of_kasner}), and that AVTD behavior holds for these solutions (Sections~\ref{s.avtd} and \ref{s.proof_main_result}).
In order to prove the AVTD property, we first prove existence of solutions of the singular initial value problem for the nonlinear velocity term dominated system.
This system is written out in Section~\ref{s.AVTD_spacetimes} and the existence of solutions is proven in Section~\ref{s.avtd}.
Finally, the AVTD property of the solutions allows us to show that the Kretschmann scalar is unbounded at each spatial point in the contracting direction, and thus, that for any such solution the spacetime is $C^2$-inextendible.

To provide specific examples which illustrate the unstable behavior of perturbed Kasner solutions with $-1<K<3$ and $K\not=1$, in Section~\ref{sec:spatially_homogeneous} we restrict to the spatially homogeneous subclass of the polarized $\Tbb^2$-symmetric spacetimes. 
We demonstrate that for $K$ in the above interval, there exist perturbations within the spatially homogeneous subclass whose asymptotic behavior as $t \searrow 0$ depends discontinuously on the perturbing parameter. 

As we were completing the writeup of this work, the paper \cite{fournodavlos2020b}, in which similar results for the stability of Kasner solutions is proved, was posted to the arXiv preprint server. 
The far-reaching work of \cite{fournodavlos2020b} treats the Einstein--vacuum system in a large number dimensions, the Einstein--scalar field system, and the polarized $U(1)$-symmetric Einstein--vacuum equations in $3+1$ dimensions. 
The latter class of spacetimes includes the polarized $\Tbb^2$-symmetric spacetimes as a subset. 
However, the formulation of the Einstein equations that is specifically adapted to the $\Tbb^2$-symmetry, as we employ in this work with the areal gauge, provides interesting insights that are not visible in the more general setting and the orthonormal frame formalism used in the impressive work of \cite{fournodavlos2020b}.
In particular, we obtain detailed asymptotic estimates that recover the full leading order term which is identified in an analysis of the singular initial value problem for the polarized $\Tbb^2$-symmetric Einstein-vacuum equations (cf. Remark~\ref{rem.sharp_asymptotics}).
Further, our formulation is sensitive to the instability of Kasner solutions for which $-1<K<3$ and $K\not=1$ (cf.~Remark~\ref{rem:instability} and Section~\ref{sec:spatially_homogeneous}). 
This geometrically distinguished property of polarized $\Tbb^2$-symmetric vacuum perturbations of Kasner solutions is a pure gauge effect for \emph{generic} perturbations in the (larger class of) polarized $U(1)$-symmetric vacuum solutions considered in \cite{fournodavlos2020b}. 
Despite the fact that the class of polarized $U(1)$-symmetric solutions contains the class of polarized $\Tbb^2$-symmetric solutions, the geometric analysis of the full polarized $\Tbb^2$-symmetric setting therefore benefits from methods that are tailored to the particular symmetry, which we employ here.

\section{$\Tbb^2$-Symmetric Spacetimes and Certain Subfamilies}
\label{s.T2spacetimes}
\subsection{Areal Coordinates}
We assume a Lorentzian manifold with topology $M=I \times \Tbb^3$ for some interval $I \subset (0, \infty)$.
The class of $\Tbb^2$-symmetric spacetimes are characterized by a $\Tbb^2$ isometry group acting on $\Tbb^3$ effectively \cite{CHRUSCIEL:1990}.
The time-dependent areas of the symmetry orbits provide a useful time coordinate, which along with a coordinatization $(\theta, x, y)$ of $\Tbb^3$  are known as the areal coordinates.
In areal coordinates the metrics for this family of spacetimes can be written in the form
\begin{align}
\label{T2metric}
    g = e^{2(\nu - u)} \left( -\alpha\rmd t^2 + \rmd \theta^2 \right)
    + e^{2u}\left(\rmd x + Q \rmd y +  (G + QH) \rmd \theta \right)^2
    + e^{-2u} t^2\left(\rmd y + H \rmd \theta\right)^2,
\end{align}
where the fields $\nu, u, Q, \alpha, G, H$ are functions of $t, \theta$ only. 
Global existence of solutions on the time interval $I=(0,\infty)$ to the Einstein vacuum equations in this symmetry class in areal coordinates is shown in \cite{BERGER1997,LeFloch:2016}.

The Killing vector fields associated to the $\Tbb^2$ symmetry (constant linear combinations of $\partial_x, \partial_y$) give rise to twist quantities\footnote{Here $\star$ denotes the Hodge star operator, and $(\del{x})^\flat = g(\del{x}, \cdot)$ is the one form which is metric dual to the vector field $\del{x}$.} 
\begin{align*}
  J_{\del{x}} =& \star (\mathrm{d}(\del{x})^\flat\wedge(\del{x})^\flat\wedge (\del{y})^\flat)
      = t e^{4u - 2\nu} \alpha^{-1/2} (\del{t} G + Q \del{t}H), \AND \\
  J_{\del{y}} =& \star (\mathrm{d}(\del{y})^\flat\wedge (\del{x})^\flat\wedge (\del{y})^\flat)
      = -t e^{-2\nu} \alpha^{-1/2} (t^2 \del{t} H + Q e^{4u}(\del{t} G + Q \del{t}H)),
\end{align*} 
that can be shown to be constant \cite{Geroch:1972jmp}, as a consequence of Einstein's vacuum equations.
In general $\Tbb^2$-symmetric spacetimes a suitable combination of $\partial_x, \partial_y$ can be taken such that one of these twist constants vanishes \cite{Berger:2019}. 
An important subclass known as the Gowdy spacetimes \cite{Gowdy1974} is characterized by the vanishing of both twist constants.

\subsection{The Polarized $\Tbb^2$-Symmetric Spacetimes}
A $\Tbb^2$-symmetric spacetime is called polarized if there exist spacelike Killing vector fields $X, Y$ spanning the $\Tbb^2$ Lie algebra such that $Q = g(X,Y)/g(X,X) = \mathrm{const}$. 
As discussed in detail in \cite{Berger:2019}, the coordinates can be adapted\footnote{In our most recent work \cite{ABIO2021} we however exploit the fact that the set of polarized $\Tbb^2$-symmetric perturbations of Kasner solutions is slightly more general when the gauge condition that $Q$ and $J_{\del{x}}$ vanish simultaneously is \emph{not} imposed.} such that both the twist constant $J_{\del{x}}$ and the constant $Q$ vanish.
Assuming such a choice of coordinates, the metrics take the form
\begin{align}
\label{polT2metric}
    g = e^{2(\nu - u)} \left( -\alpha\rmd t^2 + \rmd \theta^2 \right)
    + e^{2u}\left(\rmd x +  G \rmd \theta \right)^2
    + e^{-2u} t^2\left(\rmd y + H \rmd \theta\right)^2,
\end{align}
where, as a consequence of the coordinate conditions, the function $G$ is time independent.
For this choice of coordinates $\rmd (\partial_x)^\flat = 2 \rmd u \wedge (\partial_x)^\flat$, and thus the twist one-form $\star (\mathrm{d}(\del{x})^\flat\wedge(\del{x})^\flat)$ associated with the Killing vector field $\del{x}$ also vanishes; the Killing vector field $\del{x}$ is therefore hypersurface orthogonal. 
The other Killing vector field $\del{y}$ is hypersurface orthogonal if and only if $J_{\del{y}}=0$.
In fact, one can show that $\del{x}$ is the unique hypersurface orthogonal member of the Killing Lie algebra spanned by $\del{x}$ and $\del{y}$ if $J_{\del{y}}\neq 0$, while all members of this Lie algebra are hypersurface orthogonal if $J_{\del{y}}=0$.
Below, in the Einstein equations we label the square of the remaining non-vanishing twist $J_{\del{y}}$ by $m$. 

The vacuum Einstein equations for polarized $\Tbb^2$-symmetric spacetimes take the form (see \cite{BERGER1997})
\begin{align}
    \del{tt} u + t^{-1}\del{t}u - \alpha \del{\theta \theta} u 
        & = \frac 12 \del{\theta} \alpha \del{\theta} u 
            + \frac 12 \alpha^{-1} \del{t}\alpha \del{t}u \label{Evac.u.tt} \\
    \del{tt} \nu - \alpha \del{\theta \theta} \nu 
        & = \frac 12 \del{\theta}\alpha\del{\theta}\nu 
            + \frac 12 \alpha^{-1} \del{t}\alpha\del{t}\nu
            - \frac 14 \alpha^{-1}(\del{\theta}\alpha)^2 
            + \frac 12 \del{\theta\theta} \alpha 
            - (\del{t}u)^2 + \alpha (\del{\theta}u)^2
            + \frac 34 m \alpha t^{-4} e^{2\nu} \label{Evac.nu.tt} \\
    \del{t}\nu 
        &= t (\del{t}u)^2 + t \alpha (\del{\theta}u)^2
            + \frac 14 m \alpha t^{-3} e^{2\nu} \label{Evac.nu.t} \\
    \del{\theta}\nu         
        &= 2t\del{t}u\del{\theta}u 
        - \frac 12 \alpha^{-1}\del{\theta}\alpha \label{Evac.nu.theta} \\
    \del{t}\alpha 
        &= - m \alpha^2 t^{-3} e^{2\nu} \label{Evac.alpha.t} \\
    \del{t}G &= 0 \label{Evac.G.t} \\
    \del{t}H &= \sqrt{m}\sqrt{\alpha} t^{-3} e^{2\nu}.   \label{Evac.H.t}
\end{align}
Equation \eqref{Evac.nu.tt} is redundant as it can be generated from the remaining equations. 
Below we work with the first order evolution equation for $\nu$ \eqref{Evac.nu.t} and ignore \eqref{Evac.nu.tt}. 

\subsection{The Kasner Spacetimes}
\label{s.kasner_spacetimes}
The well-known Kasner spacetimes \cite{kasner1921} are an important example of a family of spatially homogeneous solutions of the Einstein vacuum equations. 
These can be written in constant mean curvature gauge in the form $g_K = -\rmd \tau^2 + \tau^{2p_1} \rmd x^2 + \tau^{2p_2} \rmd y^2 + \tau^{2p_3} \rmd z^2$, where the Kasner exponents satisfy the relations $\sum_i p_i = \sum_i p_i^2 = 1$.
Such spacetimes form a subclass of the polarized $\Tbb^2$-symmetric spacetimes in which the twist vanishes $(m = 0)$.
In the areal coordinates used in \eqref{polT2metric} the fields take the form 
\begin{equation}
    \label{kasner.fields}
    u^{(K)}=\frac 12(1-K)\ln t,\quad
    \nu^{(K)}=\frac 14(1-K)^2\ln t,\quad
    \alpha^{(K)}=1, \quad
    G^{(K)}= H^{(K)} = 0
\end{equation}
for an arbitrary real constant $K$, which parameterizes the full family of Kasner spacetimes. 
The corresponding \emph{Kasner exponents} are
\begin{equation}
  \label{eq:Kasnerexponents}
  p_1=\frac{K^2-1}{K^2+3},\quad p_2=\frac{2(1+K)}{K^2+3},\quad p_3=\frac{2(1-K)}{K^2+3}.
\end{equation}

\section{Definition of AVTD Behavior}
\label{s.AVTD_spacetimes}
As described above in the Introduction~\ref{s.introduction}, solutions to the Einstein equations may or may not be asymptotically velocity term dominated (AVTD).
This important property is conjectured to hold generally in the polarized $\Tbb^2$-symmetric class.
In this subsection we define the AVTD property, and review known results regarding AVTD behavior within the $\Tbb^2$-symmetric class.

The key feature of an AVTD solution is that near the singularity, which we assume here is located at $t=0$, as in the areal gauge above, the dynamics are modeled by a simpler system of equations. 
This model system, which we refer to as the \emph{Velocity Term Dominated (VTD) System}, is obtained from the Einstein equations by dropping certain terms \cite{Isenberg:1990}, notably those involving spatial derivatives in a specified gauge.
For the polarized $\Tbb^2$-symmetric spacetimes in areal gauge the VTD system is 
\begin{align}
  t\del{t}(t\del{t} u) 
      & = -\frac m2 \alpha (t\del{t}u)t^{-2}e^{2\nu}, \label{vtd.u} \\
  t\del{t}\alpha 
      &= - m \alpha^2 t^{-2} e^{2\nu}, \label{vtd.alpha} \\
  t\del{t}\nu 
      &= (t\del{t}u)^2 + \frac m4 \alpha t^{-2} e^{2\nu}, \label{vtd.nu} \\
  \del{\theta}\nu         
      &= 2t\del{t}u\del{\theta}u 
      - \frac 12 \alpha^{-1}\del{\theta}\alpha, \label{vtd.constraint} \\
  \del{t}G &= 0, \label{vtd.G} \\
  \del{t}H &= \sqrt{m}\sqrt{\alpha} t^{-3} e^{2\nu}.   \label{vtd.H}      
\end{align}
Equations \eqref{vtd.G} and \eqref{vtd.H} are the same as the Einstein equations \eqref{Evac.G.t} and \eqref{Evac.H.t}, since they do not contain spatial derivative terms. 
The constraint \eqref{vtd.constraint}, which contains only spatial derivative terms, is also the same as \eqref{Evac.nu.theta} in the Einstein equations. 
Equations \eqref{vtd.u}-\eqref{vtd.constraint} form the main VTD system and the equations \eqref{vtd.G} and \eqref{vtd.H} can be integrated once solutions to the main VTD system are obtained. 

In this work we use the following definition of the AVTD property.
\begin{Def}
  \label{def.avtd}
  A solution $(u, \nu, \alpha, G, H)$  with twist $m >0$ of the polarized $\Tbb^2$-symmetric Einstein vacuum equations \eqref{Evac.nu.tt}-\eqref{Evac.H.t} is AVTD in areal gauge provided there exists a solution 
  \begin{equation*}
      (u^{(VTD)}, \nu^{(VTD)}, \alpha^{(VTD)}, G^{(VTD)}, H^{(VTD)})
  \end{equation*} of the VTD system \eqref{vtd.u}-\eqref{vtd.H} such that for an appropriate $k \in \Zbb$ and $\beta \in \Rbb^+$, 
  \begin{align}
    \label{def.avtd.estimates}
    \begin{split}
    \lim_{t\searrow 0}&\Bnorm{(u, \nu, \alpha, G, H) - (u^{(VTD)}, \nu^{(VTD)}, \alpha^{(VTD)}, G^{(VTD)}, H^{(VTD)})}_{H^k} = 0, \\ \AND \\
    \lim_{t\searrow 0}&\Bnorm{t^\beta\del{t}\left((u, \nu, \alpha, G, H)-(u^{(VTD)}, \nu^{(VTD)}, \alpha^{(VTD)}, G^{(VTD)}, H^{(VTD)})\right)}_{H^k} = 0,
    \end{split}
  \end{align}
  where $H^k$ denotes the Sobolev space with index $k$.
\end{Def}
The weight $t^\beta$ in the norm for the time-derivatives is motivated by the singular nature of certain metric fields. 
Roughly speaking, if a spacetime solution is AVTD in the sense specified by Definition~\ref{def.avtd}, then for some time independent functions $\kt$ and $\ut$, the field $u \sim \frac 12(1-\kt)\ln(t) + \ut $ near $t\searrow 0$.
While this asymptotic behavior provides enough information to control the difference $u - u^{(VTD)}$, without employing additional high order estimates, the AVTD behavior does not control the asymptotic difference between $\del{t}u$ and $\del{t}u^{(VTD)}$.
Below, in Section~\ref{s.proof_main_result}, we prove that AVTD behavior with $\beta=1$ holds for polarized $\Tbb^2$-symmetric spacetimes for initial data sufficiently close to Kasner initial data.

This definition emphasizes an essential point of the AVTD property which is that the solution of the full Einstein equations with this property asymptotically approaches a solution of a system of ordinary differential equations (the VTD system).
This definition is equivalent to the original definition of Isenberg and Moncrief \cite{Isenberg:1990} provided one uses suitably weighted norms for the spatial metric and the second fundamental form in that work.

We note that any AVTD solution has the feature that at each spatial point the metric fields converge to values associated with a spatially homogeneous Kasner solution. 
The particular member of the Kasner family generally varies from one spatial point to another.

Within the $\Tbb^2$-symmetric class that we focus on in this work, families of AVTD solutions were first proved to exist for polarized spacetimes in the analytic category \cite{isenberg1999}.
These results follow from the analysis of an appropriate singular initial value problem in which solutions of the VTD system are prescribed as singular data.
This approach was later extended to the so-called ``half-polarized spacetime solutions" (with or without a cosmological constant \cite{Clausen2007}), which are characterized by certain restrictions on the asymptotic data.
The existence of families of smooth \cite{ClausenThesis2007}, and Sobolev-regular \cite{ames2013a} AVTD solutions are also known within the polarized and half-polarized $\Tbb^2$-symmetric class, again making use of Fuchsian methods for analyzing the singular initial value problem.
While the above cited works establish the existence of infinite-dimensional familes of AVTD solutions in the polarized $\Tbb^2$-symmetric class, Theorem~\ref{thm:mainresult} below shows, for the first time, that AVTD solutions form an open set within this class.

\section{Main Result}
\label{s.main_result}
The purpose of this paper is to discuss and prove the following result regarding the contracting asymptotics of polarized $\Tbb^2$-symmetric vacuum solutions.

\begin{thm}
  \label{thm:mainresult}
  Pick any $K\in \Rbb$ such that $|K-1| > 2$, any twist constant $m \ge 0$, $k\in \Zbb_{\ge 3}$ and $\sigma\in (0,2\kappa_0/3)$, where $\kappa_0 = \min\{1, \frac 14(K-3)(K+1)\}$.
  Then there exists $T_0 > 0$ and $R_0 >0$ sufficiently small such that any choice of functions $(\mathring{\omega}, \mathring{\nu}, \mathring{G}, \mathring{H})\in H^k(\Tbb, \Rbb^4)$ and $(\mathring{u},\mathring\alpha)\in H^{k+1}(\Tbb, \Rbb^2)$ satisfying the Einstein constraint equation at $t=T_0$
  \begin{equation*}
    \del{\theta}\mathring{\nu}         
        = 2T_0\mathring{\omega}\del{\theta}\mathring{u} 
        - \frac 12 \mathring{\alpha}^{-1}\del{\theta}\mathring{\alpha}
  \end{equation*}
  and such that 
\begin{equation}
\label{eq:mainsmallness}
\delta:=\|(T_0\mathring{\omega}-\frac 12(1-K), T_0 \del{\theta}\mathring{u},
\mathring{\alpha}-1,T_0^{-1}e^{\mathring{\nu}},T_0\del{\theta}\mathring{\alpha})\|_{H^k}
< R_0, 
\end{equation}
generates a classical solution $(u, \nu, \alpha, G, H)$ 
to the Einstein vacuum equations of the form \eqref{polT2metric} with regularity
$(u, \partial_t u, \partial_\theta u, \nu, \partial_\theta\nu, \alpha, \partial_\theta\alpha, G, H)\in C^0\bigl((0,T_0],H^k(\Tbb,\Rbb^9)\bigr)\cap L^\infty\bigl((0,T_0],H^k(\Tbb,\Rbb^9)\bigr)\cap C^1\bigl((0,T_0],H^{k-1}(\Tbb,\Rbb^9)\bigr)$
with initial data $(u,\partial_t u,\nu,\alpha, G, H) =(\mathring{u}, \mathring{\omega}, \mathring{\nu}, \mathring{\alpha}, \mathring{G}, \mathring{H})$ at $t=T_0$. The solution satisfies the following properties:
  \begin{enumerate}
    \item Contracting asymptotics: 
    There exist asymptotic data functions $(\kt, \ut, \nut, \alphat, \Ht) \in H^{k-1}$ satisfying the asymptotic constraint
    \begin{equation*}
      \del{\theta}\nut - (1 - \kt)\del{\theta}\ut + \frac 12 \del{\theta}\ln(1+\alphat) = 0,
    \end{equation*}
    such that\footnote{We write $f(t)\lesssim g(t)$ for arbitrary two functions $f$ and $g$ defined on $(0,T_0]$ if there exists a constant $C$ such that $f(t)\le C g(t)$ for all $t\in (0,T_0]$.}
    \begin{align*}
      \Bnorm{u(t)-u^{(\kt)}(t)-\ut}_{H^{k-1}}
        &\lesssim t+t^{2\kappa_0-2\sigma},\\
      \Bnorm{t\del{t}u(t)- t\del{t} u^{(\kt)}(t)}_{H^{k-1}}
        &\lesssim t+t^{2\kappa_0-2\sigma},\\
      \Bnorm{\alpha(t) -\alpha^{(\kt)}(t) - \alphat}_{H^{k-1}}
        &\lesssim t+t^{2\kappa_0-2\sigma},\\
      \norm{\nu(t)-\nu^{(\kt)}(t)-\nut}_{H^{k-1}}&\lesssim t+t^{2\kappa_0-2\sigma},\\      
      \Bnorm{H(t)-\Ht}_{H^{k-1}}&\lesssim t^{\frac 12 \min_{\theta\in [0,2\pi)}\{(\kt(\theta)-3)(\kt(\theta)+1)\}-2\sigma},
    \end{align*}
    for all $t \in (0, T_0]$, where $(u^{(\kt)},\nu^{(\kt)},\alpha^{(\kt)}, G^{(\kt)}, H^{(\kt)})$ is the ``pointwise Kasner spacetime'' determined by $\kt$ via \eqref{kasner.fields}. 
    Lastly we have $G=\mathring{G}$ for all $(t,\theta) \in (0,T_0]\times\Tbb$ and
    the Kasner parameters $\kt$ and $K$ are related as
    \begin{equation}
      \label{eq:EstimateKasnerParameters}
      \norm{\kt-K}_{H^{k-1}}\lesssim \delta.
    \end{equation}
    \item AVTD: Each solution is asymptotically velocity term dominated in the sense of Definition \ref{def.avtd} with $\beta = 1$, and $C^2$-inextendible in the contracting direction. 
    \item Stability of the Kasner: 
      We have
    \begin{align}
      \label{thm.mainresult.kasnerstability.first}
      \Bnorm{(u-u^{(K)})/\ln t}_{H^{k-1}} 
        &\lesssim \frac 1{|\ln T_0|}\Bnorm{\mathring{u}- \frac 12(1-K)\ln T_0}_{H^{k-1}} \nonumber \\ 
        & + \Bnorm{T_0\mathring{\omega}
          -\frac 12(1-K)}_{H^{k-1}}
          +T_0+T_0^{2\kappa_0-2\sigma},\\
      \Bnorm{t\del{t}u-t\del{t}u^{(K)}}_{H^{k-1}} 
        &\lesssim \Bnorm{T_0\mathring{\omega}-\frac 12(1-K)}_{H^{k-1}}
          +T_0+T_0^{2\kappa_0-2\sigma},\\
      \norm{(\nu-\nu^{(K)})/\ln(t)}_{H^{k-1}}&\lesssim \Bnorm{T_0\mathring{\omega}-\frac 12(1-K)}_{H^{k-1}}+\Bnorm{T_0\mathring{\omega}-\frac 12(1-K)}_{H^{k-1}}^2\nonumber\\
        &+\frac 1{|\ln T_0|}\norm{\mathring{\nu}}_{H^{k-1}}
          +T_0+T_0^{2\kappa_0-2\sigma},\\ 
      \Bnorm{\alpha-\alpha^{(K)}}_{H^{k-1}}
        &\lesssim \Bnorm{\mathring{\alpha} - 1}_{H^{k-1}}
        +T_0+T_0^{2\kappa_0-2\sigma},\\
      \Bnorm{G-G^{(K)}}_{H^{k}} 
        &=\Bnorm{\mathring G}_{H^{k}},\\
        \label{thm.mainresult.kasnerstability.last}
      \Bnorm{H-H^{(K)}}_{H^{k-1}}
        &\lesssim \Bnorm{\mathring{H}}_{H^{k-1}}+T_0^{\frac 12 \min_{\theta\in [0,2\pi)}\{(\kt(\theta)-3)(\kt(\theta)+1)\}-2\sigma},
    \end{align}
    where $(u^{(K)},\nu^{(K)},\alpha^{(K)}, G^{(K)}, H^{(K)})$ is the Kasner solution determined by $K$ via \eqref{kasner.fields}.
  \end{enumerate}
\end{thm}

We remark that the constants $T_0$ and $R_0$ as well as all implicit constants in the estimates in this theorem depend on the choices of parameters $K$, $m$, $k$ and $\sigma$. The proof of Theorem~\ref{thm:mainresult} is given in Section~\ref{s.main_proof}.

\begin{rem}
  We anticipate that this result will be useful in the proof of strong cosmic censorship within the polarized $\Tbb^2$-symmetric class of spacetimes, where it is expected that all solutions are AVTD. 
  To establish strong cosmic censorship one must also show that the spacetime is geodesically complete in the expanding direction. 
  Such a result has been achieved for an open family of polarized $\Tbb^2$-symmetric spacetimes in work of LeFloch and Smulevici \cite{LeFloch:2016}.
  There is however a mismatch in the current results concerning the expanding direction versus the contracting. 
  While we require initial data imposed at a time sufficiently near the $t=0$ singularity, the results of \cite{LeFloch:2016} in the expanding direction require that the initial data be imposed sufficiently far in the future. 
  It is certainly of interest to see if results in these two limits can be brought together. 
\end{rem}

\begin{rem}
  One of the main structural restrictions of our theorem is
  the \emph{smallness assumption} \eqref{eq:mainsmallness} for the initial data imposed at $t=T_0>0$. 
  Essentially, our theorem states that for every choice of polarized $T^2$-symmetric Cauchy data satisfying the Einstein constraint for which this condition holds, there is a unique solution of the $\mathbb T^2$-symmetric polarized Einstein vacuum equations defined globally on $(0,T_0]\times\mathbb T^3$ for which the properties (1), (2) and (3) are satisfied.
  The solution in particular converges to the \emph{pointwise Kasner solution} $(u^{(\kt)}, \nu^{(\kt)}, \alpha^{(\kt)}, G^{(\kt)}, H^{(\kt)})$ given by \eqref{kasner.fields} at each spatial point $\theta$ for some spatially dependent function $\kt$. 
  This function $\kt$ is close to the constant value $K$ of the original background Kasner solution according to \eqref{eq:EstimateKasnerParameters}. 
  The particular purpose of property (3) -- Kasner stability -- is to provide estimates for the difference between any such solution of this theorem and the background Kasner solution determined by $K$. 
  The estimates in (3) interestingly suggest that the solutions which are compatible with the hypothesis of this theorem are not necessarily all uniformly close to that Kasner solution. 
  This is so because the quantities on the right sides of the estimates in (3) are not all necessarily small if \eqref{eq:mainsmallness} holds. 
  Nevertheless, it is clear that the set of all solutions, for which these right-hand sides in (3) \emph{are} small, do indeed satisfy the hypothesis of the theorem, especially \eqref{eq:mainsmallness}.  
\end{rem}

\begin{rem}
  \label{rem:instability}
The other main restriction of this theorem is the condition $|K-1|>2$ which restricts the class of Kasner solutions for which we prove stability within the family of polarized $\mathbb T^2$-symmetric vacuum solutions. 
This is consistent with earlier results \cite{isenberg1999,uggla2003,heinzle2012a,ames2013a} that strongly suggest that this restriction for $K$  is sharp. Kasner solutions with $-1< K< 3$ and $K\not= 1$ are expected to be \emph{unstable} in the contracting time direction if $m>0$ in the sense that polarized $\Tbb^2$-symmetric perturbations of Kasner solutions given by some $K$ in that range do \emph{not} converge to Kasner solutions given by some similar value $\kt$ at $t=0$ as in our theorem. 
The properties of the borderline cases $K=3$ or $K=-1$ as well as those of the flat Kasner case $K=1$ are far less clear. 
The dynamics of this instability are discussed in Section~\ref{sec:spatially_homogeneous} using explicit examples. 
In this remark here, we clarify the geometric meaning of the stability restriction $|K-1|>2$. 

Recall that each diagonal Kasner solution of the form given by \eqref{kasner.fields} and \eqref{polT2metric} and some $K\in\Rbb$ has six locally isometric diagonal coordinate representations \cite{wainwright1997}, each of which gives rise to a different value of $K$. 
Each of these six isometric representations of the same Kasner solution is obtained by taking one of the six possible permutations of the three spatial coordinates, redefining the time coordinate and finally introducing a new value of $K$ to bring the metric back to the same structure as the original metric \eqref{kasner.fields} and \eqref{polT2metric}. 
In the following we refer to any such gauge transformation as a \emph{Kasner transformation}. 
For the polarized $\Tbb^2$-symmetric perturbation problem, we consider \emph{generic} (in general spatially inhomogeneous) polarized $\Tbb^2$-symmetric perturbations of the Kasner solution given by an arbitrary $K\in\Rbb$. 
Using Kasner transformations, we can, in principle, map any polarized $\Tbb^2$-symmetric perturbation of the Kasner solution given by $K$ to another (isometric) polarized $\Tbb^2$-symmetric perturbation of the same, up to isometry, Kasner solution given by some other $K$.
Since we have fixed the gauge to represent the class of polarized $\Tbb^2$-symmetric geometries in Section~\ref{s.T2spacetimes}, however, this can only work for those Kasner transformations which are compatible with these gauge choices. 
It turns out that this rules out all Kasner transformations. 
This is so because, first, the $x$-$y$-coordinate plane coincides with the geometrically distinguished Killing orbits in our gauge and therefore all Kasner transformations except the one which swaps the $x$- and the $y$-coordinates are incompatible. 
However, the Killing field $\del{x}$ is geometrically distinguished from $\del{y}$ as the unique, up to rescaling, {hypersurface orthogonal} Killing field in our gauge if $m>0$.
Thus, the Kasner transformation that swaps the $x$ and the $y$ coordinates is in fact also incompatible. 
This shows that with respect to the gauge choices made in Section~\ref{s.T2spacetimes}, the range $|K-1|>2$ \emph{geometrically} distinguishes the regime of stable polarized $\Tbb^2$-symmetric Kasner perturbations from the range $-1<K<3$, $K\not=1$ of the (as suggested by Section~\ref{sec:spatially_homogeneous}) unstable ones.
\end{rem}

\begin{rem}  
Keeping in mind Remark~\ref{rem:instability}, we now present a \emph{fully geometric} characterization of the classes of stable and of unstable polarized $\Tbb^2$-symmetric perturbations,
which is equivalent to the gauge dependent characterization of stability in Theorem~\ref{thm:mainresult} together with the strong evidence for unstable dynamics in Section~\ref{sec:spatially_homogeneous}.
Pick an arbitrary spatially homogeneous Bianchi I solution $g^{(0)}$ on $\Rbb\times\Tbb^3$ of the vacuum equations (i.e., a Kasner solution). 
Let $q_1$, $q_2$ and $q_3$ denote the eigenvalues of the Weingarten map (i.e., the mixed component second fundamental form) induced by $g^{(0)}$ on the foliation of spatially homogeneous surfaces. 
Suppose that $q_1$, $q_2$ and $q_3$ are distinct and have corresponding eigenvector fields $X_1$, $X_2$ and $X_3$. 
It then follows that we can find $\tau\in\Rbb$ such that $p_1=-\tau q_1$, $p_2=-\tau q_2$, $p_3=-\tau q_3$ satisfy the Kasner conditions $\sum_i p_i = \sum_i p_i^2 = 1$. 
Consider the set of all vacuum metrics $g$ on $\Rbb\times\Tbb^3$ with the properties that, 
  (i), $X_2$ and $X_3$ are Killing vector fields of $g$, 
  (ii), $g(X_2,X_3)=0$, and, 
  (iii), $X_3$ is hypersurface orthogonal with respect to $g$.
Theorem~\ref{thm:mainresult} states that there is an open subset of \emph{stable} polarized $\Tbb^2$-symmetric perturbations $g$ of $g^{(0)}$ provided $1>p_1>p_2$. 
On the other hand if $1>p_2>p_1$,  Section~\ref{sec:spatially_homogeneous} provides evidence that there is an open subset of \emph{unstable} polarized $\Tbb^2$-symmetric perturbations of $g^{(0)}$.
Note that according to \eqref{eq:Kasnerexponents}, the case $1>p_1>p_2$ corresponds to $|K-1|>2$, while $p_1<p_2<1$ yields $-1<K<3$, $K\not=1$.
The exceptional case $p_1=p_2$ corresponds to $K=-1$ or $K=3$. 
\end{rem}

\begin{rem}
  While the geometric source of the restiction $|K-1|>2$ is described above, the condition also arises analytically in the application of our method.
  In the proof of the Theorem~\ref{thm:mainresult} in Section~\ref{s.main_proof} below, we introduce first order variables with parameter $a \in \Rbb$ (cf. \eqref{nvars.4}). 
  By adjusting $a$, we obtain, in Propositions~\ref{prop:main_existence} and \ref{prop:fullEFE}, an open ball of solutions of the full polarized $\Tbb^2$-symmetric equations containing the Kasner solution with a given parameter $K$, where $a = \frac{1}{2}(1-K)$.
  The inequality $|K-1|>2$ arises as a necessary condition to obtain the estimate \eqref{kappalbnd} on the coefficient matrices for the first order system. 
  In particular, the parameter must satisfy $|a| > 1$, which, under the above relation between $a$ and $K$, is equivalent to $|K-1|>2$. 
  
  This condition on $K$ also ensures that an initial time $T_0$ can be found such that the norm of the initial Kasner data, which scales as $T_0^{\frac 14(K-3)(K+1)}$, can be made small, and is thus compatible with the condition \eqref{eq:mainsmallness}.
  We note that given a ball of polarized $\Tbb^2$-symmetric initial data centered on a particular Kasner data with Kasner parameter $K$, any other Kasner data with parameter $\Kt$ satisfying $|\Kt-1| > |K-1|$ is also contained in the ball.
\end{rem}

\begin{rem}
  \label{rem.sharp_asymptotics}
  The estimates specified for the contracting asymptotics in (1) of Theorem \ref{thm:mainresult} are sharp in the sense that we recover the full leading order terms that are identified in the Fuchsian studies of the singular initial value problem for the polarized $\Tbb^2$-symmetric spacetimes \cite{isenberg1999,Clausen2007,ames2013a}.
  We note that our results, at least for the spatial metric, extend the asymptotics obtained in \cite{fournodavlos2020b} for polarized $\Tbb^2$-symmetric spacetimes considered as a subclass of the polarized $U(1)$-symmetric spacetimes, which are considered in that work. 
  In \cite{fournodavlos2020b} the authors obtain sharp estimates for the asymptotics of the second fundamental form, whereas asymptotic estimates for the frame components are multiplied by a suitably large constant power of $t$.
  The more detailed asymptotics we obtain in this work are likely a benefit of working within a symmetry-adapted gauge. 
\end{rem}

\section{Spatially homogeneous solutions and the nature of the instability for $-1<K<3$, $K\not=1$}
\label{sec:spatially_homogeneous}

Before we present the proof of Theorem~\ref{thm:mainresult}  in Section~\ref{s.main_proof},  we first derive the general class of spatially homogeneous solutions of \eqref{Evac.u.tt}-\eqref{Evac.H.t} explicitly in this section. 
On the one hand, this class of solutions serves as an illustration for the theoretical results in Theorem~\ref{thm:mainresult}. 
On the other hand, it allows us to shed light on the nature of perturbations of Kasner solutions with $-1<K<3$, $K\not=1$, and, in fact, show that such solutions are unstable.
Within the spatially homogeneous subclass of polarized $\Tbb^2$-symmetric spacetimes, this instability is pure gauge.
However, as we explain in Remark~\ref{rem:instability}, some of the corresponding gauge transformations are incompatible with the gauge choices made in Section~\ref{s.T2spacetimes} to represent the \emph{full} class of generally inhomogeneous {polarized $\Tbb^2$-symmetric perturbations}.
In this sense, the instability is a geometric feature of the polarized $\Tbb^2$-symmetric spacetimes, which we anticipate is also present in the spatially inhomogeneous setting.

Consider an arbitrary metric of the form \eqref{polT2metric} for which all metric functions $\nu$, $u$, $\alpha$, $G$ and $H$ are constant with respect to the spatial coordinates $(\theta,x,y)$ but may depend on $t$.
Any such metric is spatially homogeneous of Bianchi I-type \cite{wainwright1997}. 
It is a well-known fact \cite{wainwright1997} that all Bianchi I solutions of the vacuum equations are locally isometric to Kasner solutions. 
This allows us to construct the general class of spatially homogeneous solutions of \eqref{Evac.u.tt}-\eqref{Evac.H.t} by applying a suitable class of coordinate transformations to the standard diagonal Kasner metric in the areal coordinate form given by \eqref{kasner.fields} and \eqref{polT2metric}
  \begin{equation}
    \label{eq:Kasnermetricareal}
    g={\tilde t}^{\frac{K^2-1}{2}}(-\rmd {\tilde t}^2+\rmd\thetat^2)+{\tilde t}^{1-K} \rmd \xt^2+{\tilde t}^{1+K}\rmd \yt^2,
  \end{equation}
with the coordinates in equation \eqref{eq:Kasnermetricareal} labeled as $(\tilde t,\thetat,\xt,\yt)$.
A general class of local\footnote{The coordinate transformations in this section are in general incompatible with the \emph{global} $\Rbb\times \Tbb^3$-topology of the spacetime manifold. 
If we apply a ``local coordinate transformation'', we in fact first go to the universal cover $\Rbb\times\Rbb^3$, then apply the coordinate transformation on that manifold, and then finally go back to the original spacetime manifold $\Rbb\times\Tbb^3$ by making each spatial coordinate $2\pi$-periodic. This is common practice in the literature.} coordinate transformations from the coordinates $\xt^\mu=(\tilde t,\thetat,\xt,\yt)$ above to a new set of coordinates $x^\mu=(t,\theta,x,y)$ may now be introduced, where in a first step, we require that the Killing Lie algebras spanned by $(\del{\thetat}, \del{\xt}, \del{\yt})$ and $(\del{\theta}, \del{x}, \del{y})$ are isomorphic. 
This implies that the  Jacobian $M=(\partial \xt^\mu/\partial x^\nu)$ is of the form
  \begin{equation}
    \label{eq:coordinateJ}
    M=
    \begin{pmatrix}
      M_{0}(t) & 0 & 0 & 0\\
      M_1 & A_{11} & A_{12} & A_{13}\\ 
      M_2 & A_{21} & A_{22} & A_{23}\\ 
      M_3 & A_{31} & A_{32} & A_{33}
    \end{pmatrix},
  \end{equation}
where the submatrix $(A_{ik})$ is non-singular and constant both in space and time. 
A direct calculation reveals that the metric obtained from \eqref{eq:Kasnermetricareal} by this coordinate transformation is of the form \eqref{polT2metric} provided we set $M_1=M_2=M_3=0$ (in order to make the shift components zero), $A_{12}=A_{23}=A_{32}=0$ (in order to make $Q$ in \eqref{T2metric} zero), and $G$ constant in time (which is equivalent to the condition that the first twist constant vanishes and $\del{x}$ is therefore a hypersurface orthogonal Killing vector field), and we define the new areal  (with respect to the $x$-$y$-coordinate orbits) time coordinate
\begin{equation}
  \label{eq:transformedarealtime}
  t=A_{22}\sqrt{A_{33}^2{\tilde t}^2+A_{13}^2 {\tilde t}^{\frac{(K-1)^2}2}},
\end{equation}
from which we can determine the function $M_0(t)$ in \eqref{eq:coordinateJ} implicitly. 
Observe here that we can, in general, not express ${\tilde t}$ explicitly as a function of $t$. 
For this reason most functions in the following are  expressed as functions of ${\tilde t}$ instead of $t$.
Here, and in all of the following, we also assume that $A_{11}$, $A_{22}$ and $A_{33}$ are all positive.

Given these conditions, it follows from a straightforward calculation that the metric functions in \eqref{polT2metric} expressed in terms of the new coordinates $(t,\theta,x,y)$ take the following form:
\begin{align}
  \label{eq:explicitu}
  e^{2u}&=A_{22}^2 {\tilde t}^{1-K},\\
  G&=\frac{A_{21}}{A_{22}},\\
  H&=\frac{A_{11} A_{13}  {\tilde t}^{\frac{1}{2} (K-3) (K+1)}+A_{31} A_{33}}{A_{13}^2
    {\tilde t}^{\frac{1}{2} (K-3) (K+1)}+A_{33}^2},\\
  e^{2\nu}&=
    \frac{A_{22}^2 (A_{13}
A_{31}-A_{11} A_{33})^2}{A_{13}^2 {\tilde t}^{\frac{1}{2} (K-3)
            (K+1)}+A_{33}^2}{\tilde t}^{\frac 12(1-K)^2},\\
  \label{eq:explicitalpha}
\alpha&=\frac{16 \bigl(A_{13}^2 {\tilde t}^{\frac{1}{2} (K-3)
(K+1)}+A_{33}^2\bigr)^2}{A_{22}^2 (A_{13} A_{31}-A_{11}
A_{33})^2 \bigl(A_{13}^2 (K-1)^2 {\tilde t}^{\frac{1}{2} (K-3) (K+1)}+4
    A_{33}^2\bigr)^2}.
\end{align}
One can check that \eqref{eq:explicitu}-\eqref{eq:explicitalpha} is a solution of \eqref{Evac.u.tt}-\eqref{Evac.H.t} for 
\begin{equation}
  \label{eq:explicitm}
  m=\frac 14 A_{13}^2A_{22}^2A_{33}^2(K-3)^2(K+1)^2,
\end{equation}
using \eqref{eq:transformedarealtime}. 
Moreover, given any spatially homogeneous solution of \eqref{Evac.u.tt}-\eqref{Evac.H.t} for arbitrary $m\ge 0$, one can determine all parameters $K$, $A_{11}$, $A_{21}$, $A_{22}$, $A_{31}$ and $A_{33}$ and $A_{13}$ from the data of the solution at some arbitrary initial time $t=T_0$ and from \eqref{eq:explicitm} so that this solution agrees with \eqref{eq:explicitu}-\eqref{eq:explicitalpha} for all times $t$. 
The functions \eqref{eq:explicitu}-\eqref{eq:explicitalpha} therefore represent the general spatially homogeneous solution of \eqref{Evac.u.tt}-\eqref{Evac.H.t} for arbitrary $m\ge 0$.
The solution is twist-free (i.e. $m=0$) if and only if $A_{13}=0$ (unless $K=3$ or $K=-1$) which follows from \eqref{eq:explicitm}.
The original diagonal Kasner solutions \eqref{kasner.fields} correspond to the special case
\begin{equation}
  \label{eq:diagonalKasnerparameters}
  A_{11}=A_{22}=A_{33}=1,\quad A_{21}=A_{13}=A_{31}=0.
\end{equation}

Observe here that the \emph{twist is pure gauge} in the spatially homogeneous case. 
If we interpret \eqref{eq:explicitu}-\eqref{eq:explicitalpha} as a particular subclass of the family of polarized $\mathbb T^2$-symmetric solutions, however, all coordinate transformations which would take solutions with $m>0$ to solutions with $m=0$ are inconsistent with the gauge choices in Section~\ref{s.T2spacetimes}, which represent the \emph{full} class of {polarized $\Tbb^2$-symmetric geometries}.
Specifically, within the spatially homogeneous class, gauge transformations with an arbitrary value for $A_{13}$ are allowed as above, while in the full polarized $\Tbb^2$-symmetric class given in the gauge in Section~\ref{s.T2spacetimes} only gauge transformations with $A_{13}=0$ are permitted.

Given $K\in\Rbb$, and arbitrary parameters $\kt$ close to $K$, and $A_{11}$, $A_{21}$, $A_{22}$, $A_{31}$ and $A_{33}$ and $A_{13}$ close to the values in \eqref{eq:diagonalKasnerparameters},
we interpret \eqref{eq:explicitu}-\eqref{eq:explicitalpha} with $K$ replaced by $\kt$ as a (spatially homogeneous) perturbation of the original diagonal Kasner solution \eqref{kasner.fields} given by the parameter $K$. We refer to the Killing sub Lie algebra spanned by $\del{x}$ and $\del{y}$ for \eqref{eq:explicitu}-\eqref{eq:explicitalpha} as \emph{$\Tbb^2$-symmetry}, and recall that $\del{x}$ is by construction hypersurface orthogonal as required by our choice of gauge for the class of polarized $\Tbb^2$-symmetric metrics. 
As a consequence, Theorem~\ref{thm:mainresult} must apply to these perturbations.
The unstable behavior of solutions \eqref{eq:explicitu}-\eqref{eq:explicitalpha} with $\kt$ close to $K$ and $-1<K<3$ and $K\not=1$ is made manifest in the asymptotic behavior as $t\searrow 0$.
Notice that as a consequence of
\[\frac 12 (K-1)^2 - 2=\frac12 (K-3)(K+1),\]
we have, close to $t=\tilde{t}=0$, that $t=A_{22} A_{33}\tilde t+\ldots$ if $K>3$ or $K<-1$ or $A_{13}=0$, and $t=A_{22} |A_{13}|{\tilde t}^{\frac{(K-1)^2}4}+\ldots$ if $-1<K<3$, $K\not=1$ and $A_{13}\not=0$. 
In the regime $K>3$ or $K<-1$ or $A_{13}=0$, if the dynamics of polarized $\Tbb^2$-symmetric vacuum solutions is stable according to Theorem~\ref{thm:mainresult}, it therefore follows that $t$ and $\tilde t$ are ``essentially the same'' close to $t=0$ (up to a factor) and the limits $t\searrow 0$ of \eqref{eq:explicitu}-\eqref{eq:explicitalpha} can be read off directly:
\begin{gather*}
  {t}^{-(1-\kt)} e^{2u}=\frac{A_{22}^2}{(A_{22}A_{33})^{1-\kt}},\quad
  G=\frac{A_{21}}{A_{22}},\quad
  H\rightarrow\frac{A_{31}}{A_{33}},\\
  {t}^{-\frac 12(1-\kt)^2}e^{2\nu}\rightarrow
    \frac{A_{22}^2 (A_{13}
A_{31}-A_{11} A_{33})^2}{A_{33}^2 (A_{22}A_{33})^{(\kt-1)^2/2}},\quad
\alpha\rightarrow\frac{\bigl(A_{13}^2+A_{33}^2\bigr)^2}{A_{22}^2 (A_{13} A_{31}-A_{11}  A_{33})^2 A_{33}^4}.
\end{gather*}
We see that these limits depend continuously on the perturbation parameters $\kt$ close to $K$, and on $A_{11}$, $A_{21}$, $A_{22}$, $A_{31}$ and $A_{33}$ and $A_{13}$ close to the values in \eqref{eq:diagonalKasnerparameters}.

However, in the case $-1<K<3$ and $K\not=1$ we find that the limits at $t=0$
are significantly different, depending upon whether or not $A_{13}$ vanishes.
If $A_{13}=0$, then
\begin{gather*}
  t^{-(1-\kt)} e^{2u}\rightarrow \frac{A_{22}^2}{(A_{22}A_{33})^{1-\kt}},\quad
    G= \frac{A_{21}}{A_{22}},\quad
    H= \frac{A_{31}}{A_{33}},\\
    t^{-\frac{1}{2}(1-\kt)^2} e^{2\nu} \rightarrow \frac{A_{22}^2 A_{11}^2}{(A_{22}A_{33})^{(\kt-1)^2/2}},\quad
 \alpha=\frac{1}{A_{11}^2 A_{11}^2 A_{33}^2},
\end{gather*}
while if $A_{13}\not=0$, we have
\begin{gather*}
    e^{2u} t^{-\frac {4}{1-\kt}}\rightarrow \frac{A_{22}^2}{(A_{22}|A_{13}|)^{\frac {4}{1-\kt}}},\quad
    G= \frac{A_{21}}{A_{22}},\quad
    H\rightarrow\frac{A_{11}}{A_{13}},\\
    e^{2\nu} t^{-\frac {8}{(\kt-1)^2}}\rightarrow\frac{A_{22}^2 (A_{13}
  A_{31}-A_{11} A_{33})^2}{A_{13}^2 (A_{22}|A_{13}|)^{\frac {8}{(\kt-1)^2}}}, \quad
 \alpha\rightarrow\frac{16 }{A_{22}^2 (A_{13}
  A_{31}-A_{11} A_{33})^2 (\kt-1)^4}.
\end{gather*}
This shows explicitly that \emph{in the parameter range $-1<K<3$, $K\not=1$, the diagonal Kasner solutions (for which $m=0$) are unstable, and therefore they are, in a sense, bad models for general $m>0$-solutions. In fact, any arbitrarily small deviation from the twist-free case $m=0$ leads to drastically different asymptotic behavior at $t=0$.}

Given the explicit solutions discussed above, it is conceivable that 
certain perturbations of the $m\neq 0$ Kasner solutions (as opposed to the diagonal twist-free Kasner solutions) within the family of polarized $\Tbb^2$-symmetric solutions might be stable even if $-1<K<3$. 
While these Kasner solutions with twist are isometric to the diagonal twist-free Kasner solutions in the spatially homogeneous setting (as we have shown above), this would nevertheless constitute a geometrically distinguished perturbation analysis within the general class of spatially inhomogeneous polarized $\Tbb^2$-symmetric solutions.

\section{Proof of Main Result}
\label{s.main_proof}

\subsection{Reformulation of the Polarized $\Tbb^2$-Symmetric Vacuum Einstein equations as a First Order Symmetric Hyperbolic Fuchsian System} 
As noted above, to carry out the proof of our main result, Theorem~\ref{thm:mainresult}, we rely on Theorem~3.8 from reference \cite{BOOS:2020}, which we re-state in a modified form as Theorem~\ref{symthm} below. 
This theorem requires that the Einstein equations for the polarized $\Tbb^2$-symmetric spacetimes be re-expressed as a first-order symmetric hyperbolic Fuchsian system. 
We do this in two steps.  
As shown in the appendix of \cite{isenberg1999}, equations \eqref{Evac.u.tt}, \eqref{Evac.nu.t}, and \eqref{Evac.alpha.t} can be recast in first-order symmetric hyperbolic form. 
To do this, we introduce the variables 
\begin{equation}
\label{first_order_variables}
   (z_0, z_1, z_2) = (u, \del{t}u, \del{\theta}u), \quad \zeta = \del{\theta}\alpha.
\end{equation}
The resulting \emph{evolution system}, obtained as described in \cite{isenberg1999}, is
\begin{align}
\del{t}z_0&=z_1, \label{vacA.1}\\
\del{t}z_1&=\alpha \del{\theta}z_2-\frac{1}{t}z_1 - \frac{m}{2t^3}\alpha z_1 e^{2\nu}+\frac{1}{2}z_2\zeta, \label{vacA.2}\\
\alpha\del{t}z_2&=\alpha\del{\theta}z_1, \label{vacA.3}\\
\del{t}\alpha &=-\frac{m}{t^3}\alpha^2 e^{2\nu}, \label{vacA.4}\\
\del{t}\nu &= t z_1^2+t\alpha z_2^2+\frac{m}{4t^3}\alpha e^{2\nu} \label{vacA.5}
\intertext{and}
\del{t}\zeta &= -\frac{2m}{t^3}\alpha e^{2\nu}\biggl[\zeta+\alpha\biggl(2 t z_1 z_2 -\frac{\zeta}{2\alpha}\biggr)\biggr], \label{vacA.6}
\end{align}
together with \eqref{Evac.G.t} and \eqref{Evac.H.t}. Eq.~\eqref{Evac.nu.theta} and \eqref{first_order_variables} yield the following \emph{constraint equations} 
\begin{align}
  \label{eq:vacC.1}
  \del{\theta}\nu         
        &= 2t z_1 z_2 
        - \frac 12 \alpha^{-1}\zeta,\\
  \label{eq:vacC.2}
  \del{\theta}z_0&=z_2, \\
  \label{eq:vacC.3}
  \del{\theta} \alpha &= \zeta,
\end{align}
which must hold at each $t$ in addition to the above evolution equations. It has been shown in \cite{isenberg1999} that the \emph{constraint equations propagate}; i.e., given an arbitrary  time interval $I\subset (0,\infty)$ and a sufficiently smooth solution of \eqref{vacA.1}-\eqref{vacA.6} on $I$ with the property that \eqref{eq:vacC.1}-\eqref{eq:vacC.3} are satisfied at one instance of time $t_0\in I$, then \eqref{eq:vacC.1}-\eqref{eq:vacC.3} are satisfied for all $t\in I$.

While the system of evolution equations \eqref{vacA.1}-\eqref{vacA.6} is now in a useful symmetric hyperbolic form, which gives rise to the standard \emph{local-in time} well-posedness of the Cauchy problem, to apply the \emph{global-in time} existence Theorem~\ref{symthm}, it is necessary to recast the system in the \emph{Fuchsian form}.
To do this, we define the variables $\xi$, $\psi$, $w_0$, $w_1$, $w_2$ and $\eta$ via 
\begin{align}
    \nu&=\ln(t)+\ln(\xi), \label{nvars.1}\\
    \alpha&=1+\psi, \label{nvars.2} \\
    z_0&={a} w_0, \label{nvars.3}\\
    z_1&=\frac{1}{t}(a+w_1), \label{nvars.4}\\
z_2 & = \frac{1}{t}w_2 \label{nvars.5}
\intertext{and}
\zeta &= \frac{1}{t}\eta \label{nvars.6},
\end{align}
where $a\in \Rbb\setminus\{0\}$ is a constant. 
A short calculation shows we can write \eqref{vacA.1}-\eqref{vacA.6} in terms of these variables as 
\begin{align}
 \del{t}w_0&=\frac{1}t\Bigl(\frac{1}{a}w_1+1\Bigr), \label{vacB.1}  \\
 2\del{t}w_1&= 2(1+\psi)\del{\theta}w_2+\frac{1}{t}\eta w_2-\frac{1}{t}m(1+\psi)(a+w_1)\xi^2, \label{vacB.2}\\
  2(1+\psi)\del{t} w_2&=2(1+\psi)\del{\theta}w_1+\frac{2}{t}(1+\psi)w_2,
  \label{vacB.3} \\
\del{t}\psi &= -\frac{1}{t}m(1+\psi)^2\xi^2,  \label{vacB.4}\\
4 \del{t}\xi&=\frac{1}{t} \bigl(4
   a^2-4+8 a w_1+4w_2^2 (1+\psi )+4w_1^2\bigr)\xi+ \frac{1}{t} m  (1+\psi)\xi^3 \label{vacB.5}
   \intertext{and}
 \del{t}\eta &= \frac{1}{t}\eta -\frac{m}{t}(1+\psi) \bigl(4 w_2 (a+w_1) (1+\psi )+\eta \bigr)\xi^2. \label{vacB.6}
\end{align}
The constraint equations \eqref{eq:vacC.1}--\eqref{eq:vacC.3} take the form
\begin{align}
  \label{eq:vacCC.2}
  t{a}\del{\theta}w_0&=w_2,\\
  \label{eq:vacCC.1}
  t\del{\theta}\xi         
        &= 2 (a+w_1) w_2 \xi
        - \frac 12 (1+\psi)^{-1}\eta \xi,\\
  \label{eq:vacCC.3}
  t \del{\theta}\psi&=\eta.
\end{align}
We note that the equations \eqref{vacB.2}-\eqref{vacB.6} form a closed subsystem of evolution equations for the variables $\{w_1,w_2,\psi,\xi,\eta\}$. 
We therefore focus on this subsystem first, together with \eqref{eq:vacCC.1}-\eqref{eq:vacCC.3}, and then solve the decoupled equations \eqref{Evac.G.t}, \eqref{Evac.H.t} and \eqref{vacB.1}  together with \eqref{eq:vacCC.2} in order to obtain solutions of the full vacuum Einstein's equations.

Writing the core evolution system \eqref{vacB.2}-\eqref{vacB.6} in matrix form, we have
\begin{equation}
    B^0\del{t}U+B^1\del{\theta}U=\frac{1}{t}\Bc\Pbb U
    +\frac{1}{t}F \label{vacC}
\end{equation}
where
\begin{align}
    U&= (w_1,w_2,\psi,\xi,\eta)^{\tr}, \label{Udef}\\
    B^0&=\begin{pmatrix} 
     2 & 0 & 0 & 0 & 0\\
     0 & 2(1+\psi) & 0 & 0 & 0\\
     0 & 0 & 1 & 0 & 0\\
     0 & 0 & 0 & 4 & 0\\
     0 & 0 & 0 & 0 & 1\end{pmatrix}, \label{B0def}\\
    B^1&=\begin{pmatrix} 
     0 & -2(1+\psi) & 0 & 0 & 0\\
     -2(1+\psi) & 0 & 0 & 0 & 0\\
     0 & 0 & 0 & 0 & 0\\
     0 & 0 & 0 & 0 & 0\\
     0 & 0 & 0 & 0 & 0\end{pmatrix}, \label{B1def}\\
    \Bc&=\begin{pmatrix} 
     2 & 0 & 0 & 0 & 0\\
     0 & 2(1+\psi) & 0 & 0 & 0\\
     0 & 0 & 1 & 0 & 0\\
     0 & 0 & 0 & 4
   a^2-4+8 a w_1+4w_2^2 (1+\psi )+4w_1^2 +m(1+\psi)\xi^2& 0\\
     0 & 0 & 0 & -4m(a+w_1)(1+\psi)^2 w_2\xi & 1-m (1+\psi)\xi^2\end{pmatrix}, \label{Bcdef}\\
    \Pbb&=\begin{pmatrix} 
     0 & 0 & 0 & 0 & 0\\
     0 & 1 & 0 & 0 & 0\\
     0 & 0 & 0 & 0 & 0\\
     0 & 0 & 0 & 1 & 0\\
     0 & 0 & 0 & 0 & 1\end{pmatrix}, \label{Pbbdef}\\
    \intertext{and}
     F &= \begin{pmatrix}
    \eta w_2-m(1+\psi)(a+w_1)\xi^2\\0\\-m(1+\psi)^2\xi^2\\0\\0
    \end{pmatrix}. \label{Fdef}
\end{align}

The reason for rewriting equations \eqref{vacB.2}-\eqref{vacB.6} in the form of equation \eqref{vacC} is to verify that this system is in the Fuchsian form of Section~\ref{sec:Fuchsian} in the Appendix, and more importantly, as we show in the next section, that the system satisfies, for certain choices of the constants $a,b$, all of the coefficient assumptions stated in Section \ref{coeffassumps} of the Appendix. 
This establishes via an application of Theorem \ref{symthm}, the existences of solutions to the global initial value problem (GIVP)
\begin{align}
    B^0\del{t}U+B^1\del{\theta}U&=\frac{1}{t}\Bc\Pbb U
    + \frac{1}{t}F\quad \text{in $(0,T_0]\times \Tbb$,} \label{vacIVP.1}\\
    U &= \mathring{U}\hspace{2.1cm} \text{in $\{T_0\}\times \Tbb$,} \label{vacIVP.2}  
\end{align}
under a suitable smallness condition imposed on the initial data $\mathring{U}$, and for a suitable choice of the constant $a$. 

\subsection{Global Existence for the Cauchy Problem Near the Singularity for Initial Data near Kasner}
\label{s.global_existence}

\subsubsection{Coefficient properties\label{coefficient}}
To apply Theorem~\ref{symthm} to the Einstein equations for polarized $\Tbb^2$-symmetric spacetimes, we must verify that the coefficients appearing in these equations satisfy the hypotheses of Theorem~\ref{symthm}. 
We do this in this subsection.
Let $a\in\mathbb R\setminus\{0\}$ be given. Suppose that $R>0$ and\footnote{In equation \eqref{eq:defU} and below, the superscript ``tr" indicates the transpose operation.}
\begin{equation}
\label{eq:defU}
    U=(w_1,w_2,\psi,\xi,\eta)^{\tr}\in B_R(\Rbb^5),
\end{equation}
where $B_R$ is the ball of radius $R$.
Notice that the vector $U$ is at this stage just a collection of real-valued functions and not yet necessarily a solution of \eqref{vacIVP.1}-\eqref{vacIVP.2}.
First, we note that if
\begin{equation}
\label{psibndF}
    R<1,
\end{equation}
then the matrix $B^0$, defined by \eqref{B0def}, satisfies
\begin{equation} \label{B0lowbnd}
{\min\{1,2(1-R)\}}\id \leq B^0.
\end{equation}
It is also clear from \eqref{B0def} and \eqref{B1def}  that the matrices $B^0$ and $B^1$ are symmetric; that is,
\begin{equation} \label{B0B1sym}
    (B^0)^{\tr}=B^0 \AND (B^1)^{\tr}=B^1,
\end{equation}
while we see from  \eqref{B0def}, \eqref{B1def}, \eqref{Bcdef} and \eqref{Fdef} that 
$B^0$, $B^1$, $\Bc$ and $F$ are
smooth
in $(t,U)$.
In particular, this is enough to guarantee that the system \eqref{vacIVP.1} is symmetric hyperbolic, which is, in turn, enough by standard local-in-time existence and uniqueness theorems to guarantee that the initial value problem \eqref{vacIVP.1}-\eqref{vacIVP.2} has solutions on time intervals of the form $(T_1,T_0]$ for some $T_1 \in (0,T_0)$ where the initial data can be large, but typically the size of the interval of existence is small; i.e., $|T_0-T_1| \ll |T_0|$. What we are interested in is the global initial value problem (GIVP) with $T_1=0$, which in general can only hold under a small initial data assumption.

Next, we note that the matrix $\Pbb$ defined by \eqref{Pbbdef} is a constant, symmetric projection operator; that is 
\begin{equation} \label{Pbbprops}
\Pbb^2=\Pbb, \quad \Pbb^{\tr}=\Pbb \AND \del{t}\Pbb=\del{\theta}\Pbb = 0.
\end{equation}
We further observe from \eqref{B0def} and \eqref{Bcdef}
that 
\begin{equation}\label{B0Bccom}
[\Pbb,B^0]=[\Pbb,\Bc]=0.
\end{equation}
Letting
\begin{equation} \label{Pbbperp}
\Pbb^\perp = \id -\Pbb
\end{equation}
denote the complementary projection operator, it follows
immediately from \eqref{B0Bccom} that $B^0$ satisfies
\begin{equation}\label{PbbB0Pbbperp}
    \Pbb B^0 \Pbb^\perp = \Pbb^\perp B^0 \Pbb = 0.
\end{equation}
We further observe from  \eqref{Udef}, \eqref{Pbbdef}, \eqref{Fdef}, and \eqref{Pbbperp} that $F$ satisfies
\begin{equation} \label{Fbnds1}
    \Pbb F = 0 
\end{equation}
and there exists a positive constant\footnote{\label{ft.oh_notation} We use the usual big `$O$', and little `$o$' notation: In the limit $t\searrow 0$ $f(t) = O(g(t))$ if $\norm{f} < C\norm{g}$ for some constant $C$, and $f(t) = o(g(t))$ if $\lim_{t\to 0} \norm{f}/\norm{g} = 0$ where the norm is chosen according to the context. The notation $\Ordc$ is defined as follows: $f(t,v)=\Ordc(g(t,v))$ for $(t,v)\in (0,T_0]\times B_{R}(\Rbb^N)$ if $f(t,v)$ satisfies $f(t,v)= h(t,v)g(t,v)$ for some $h(t,v)$ satisfying $|h(t,v)|\leq 1$ and $|\del{v}^\ell h(t,v)| \leq C_\ell$ for all $(t,v)\in (0,T_0]\times B_{R}(\Rbb^N)$ and $\ell \in \Zbb_{\geq 1}$.} $\lambda=\Ord(R)$ such that
\begin{equation} \label{Fbnds2}
    \Pbb^\perp F = \Ordc\biggl(\frac{\lambda}{R}|\Pbb U|^2\biggr)
\end{equation}
for all $|U|<R$.

Next, from \eqref{B0def} and \eqref{Bcdef},
we see that\footnote{Here we write $B|_{U=0}$ to denote that $B$ is evaluated at the zero vector in the argument corresponding to $U$.}
\begin{equation*}
   \kappa_0 B^0|_{U=0} \leq \Bc|_{U=0}
\end{equation*}
where 
\begin{equation} \label{kappa0def}
\kappa_0 = \min\{1,a^2-1\}.
\end{equation}
So if $|a|>1$, then for any $\sigma \in (0,\kappa_0)$, there exists a $R>0$ depending on $\sigma$ such that
\begin{equation} \label{kappalbnd} 
    \kappa B^0(U) \leq \Bc(U)
\end{equation}
for all $|U|<R$
where
\begin{equation} \label{kappadef}
    \kappa=\kappa_0-\sigma.
\end{equation}

To complete the verification of the coefficient properties from Section \ref{coeffassumps}, we
observe from \eqref{vacC}-\eqref{Fdef} that
$\Div\! B$, see \eqref{DivB}, is given by
\begin{equation*}
    \Div\! B(U,\partial_\theta U) =  \begin{pmatrix} 
    0 & -2\del{\theta}\psi & 0 & 0 & 0\\
    -2\del{\theta}\psi & -\frac{2m}{t}(1+\psi)^2\xi^2 & 0 & 0 & 0\\
    0 & 0 & 0 & 0 & 0\\
    0 & 0 & 0 & 0 & 0\\
    0 & 0 & 0 & 0 & 0
    \end{pmatrix}.
\end{equation*}
From this formula, \eqref{Udef} and \eqref{Pbbdef}, it is clear that there exists a positive constant  $\Theta=\Ord(R)$ and $\beta$ such that
\begin{equation} \label{DivBbnd}
\Div B(U,\partial_\theta U) = \Ordc\Bigl(\Theta+\frac{\beta}{t}|\Pbb U|^2\Bigr)
\end{equation}
for all $|U|<R$, $|\partial_\theta U|<R$ and $t\in (0,T_0]$.

Together, the results \eqref{B0lowbnd}-\eqref{Fbnds2} and \eqref{kappalbnd}-\eqref{DivBbnd} imply that the Fuchsian
equation \eqref{vacIVP.1} satisfies all the coefficient assumption from Section \ref{coeffassumps} for the above choice of constants $R$, $\kappa$, $\lambda$, $\Theta$ and $\beta$. Moreover, we note that by choosing $R$ sufficiently small we can make the constant $\lambda$ as small as we like.

\subsubsection{Global existence}
We are now in a position to apply Theorem~\ref{symthm} from the appendix. Doing so yields, for a suitably small choice of the initial data, the existence of a classical solution  $U\in C^1((0,T_0]\times \Tbb,\Rbb^5)$ to the GIVP \eqref{vacIVP.1}-\eqref{vacIVP.2} that satisfies the bound
\begin{equation}\label{Ubnd}
    \norm{U(t)}_{L^\infty}<R
\end{equation}
for all $t\in(0,T_0]$, where $R=R(\sigma)>0$ is small enough to satisfy \eqref{psibndF} and
ensure that \eqref{kappalbnd} holds and that the inequality $\kappa>\gamma_1(\lambda+\beta/2)$ is satisfied.

Since \eqref{vacB.5} is an ODE that admits the trivial solution $\xi=0$, it follows that if  $\xi|_{t=T_0}>0$ on $\Tbb$, then this holds for all $t\in (0,T_0]$.
The following proposition contains the precise statement of the existence result. 

\begin{prop} \label{prop:main_existence}
Suppose $k \in \Zbb_{\geq 2}$, $T_0>0$, $|a|>1$ and $\sigma \in (0,\kappa_0)$ with $\kappa_0$ given in \eqref{kappa0def}, and\footnote{Here and below, $H^k$ denotes the Sobolev space encompassing k'th order derivatives, and $\|\cdot\|_{H^k}$ denotes the corresponding norm.}
\begin{equation*}
\mathring{U}=\bigl(\mathring{w}_1,\mathring{w}_2,\mathring{\psi},\mathring{\xi},\mathring{\eta}\bigr)^{\tr}\in H^k(\Tbb,\Rbb^5)
\end{equation*}
is chosen so that $\mathring{\xi}>0$ in $\Tbb$.
Then there exists a constant 
$R_0>0$ such that, if $\mathring{U}$ is chosen to satisfy
\begin{equation}
  \label{eq:CDSmallness}
 \norm{\mathring{U}}_{H^k}< R_0,
\end{equation}
then there exists a unique solution 
\begin{equation*}
U \in C^0\bigl((0,T_0],H^k(\Tbb,\Rbb^5)\bigr)\cap L^\infty\bigl((0,T_0],H^k(\Tbb,\Rbb^5)\bigr)\cap C^1\bigl((0,T_0],H^{k-1}(\Tbb,\Rbb^5)\bigr)\subset C^1\bigl((0,T_0]\times\Tbb,\Rbb^5\bigr)
\end{equation*}
of the GIVP \eqref{vacIVP.1}-\eqref{vacIVP.2} such that the limit $\lim_{t\searrow 0} \Pbb^\perp U(t)$, denoted $\Pbb^\perp U(0)$, exists in $H^{k-1}(\Tbb,\Rbb^5)$.
Moreover $U$ satisfies the bound
\eqref{Ubnd}
on $(0,T_0]\times \Tbb$, and the component $\xi$ of $U$ satisfies $\xi>0$ in $(0,T_0]\times \Tbb$.
Finally, for $0<t<T_0$,  the solution $U$ satisfies  the energy estimate
\begin{equation}  \label{eq:energyestimates}
\norm{U(t)}_{H^k}^2+ \int_t^{T_0} \frac{1}{\tau} \norm{\Pbb U(\tau)}_{H^k}^2\, d\tau   \lesssim \norm{\mathring{U}}_{H^k}^2
\end{equation}
and the decay estimates
\begin{equation}
  \label{eq:decayestimates}
\norm{\Pbb U(t)}_{H^{k-1}} \lesssim 
t^{\kappa_0-\sigma}
\AND \norm{\Pbb^\perp U(t) - \Pbb^\perp U(0)}_{H^{k-1}} \lesssim t+t^{2\kappa_0-2\sigma}.
\end{equation}
\end{prop}
Observe that, in contrast to Theorem~\ref{s.main_result}, we do not require $T_0$ to be small here. The smallness requirement for $T_0$ is introduced only in Section~\ref{s.stability_of_kasner}.

\subsection{Improved Decay Estimates for the Solutions Approaching the Singularity}
\label{ss.improved_asymptotics}
Having established the global existence of the GIVP \eqref{vacIVP.1}-\eqref{vacIVP.2} for sufficiently small initial data, the next step is to improve the information regarding the asymptotic behavior for the solutions (as $t \searrow 0$) in \eqref{eq:decayestimates} (at the cost of an order of differentiability).

\begin{prop} \label{prop:hoexp}
Suppose $k \in \Zbb_{\geq 3}$, $T_0>0$, $|a|>1$, and $\sigma \in (0, 2\kappa_0/3)$ with $\kappa_0$ given in \eqref{kappa0def}, and
\begin{equation}
\label{eq:HOCD}
\mathring{U}=\bigl(\mathring{w}_1,\mathring{w}_2,\mathring{\psi},\mathring{\xi},\mathring{\eta}\bigr)^{\tr}\in H^k(\Tbb,\Rbb^5)
\end{equation}
satisfies $\norm{\mathring{U}}_{H^k}< R_0$
for $R_0>0$ small enough so that by Proposition~\ref{prop:main_existence} there
exists a unique solution
\begin{equation*}
U=\bigl(w_1,w_2,\psi,\xi,\eta\bigr)^{\tr} \in C^0\bigl((0,T_0],H^k(\Tbb,\Rbb^5)\bigr)\cap L^\infty\bigl((0,T_0],H^k(\Tbb,\Rbb^5)\bigr)\cap C^1\bigl((0,T_0],H^{k-1}(\Tbb,\Rbb^5)\bigr)
\end{equation*}
to the GIVP \eqref{vacIVP.1}-\eqref{vacIVP.2} such that the limit $\lim_{t\searrow 0} \Pbb^\perp U(t)$ exists in $H^{k-1}(\Tbb,\Rbb^5)$, $U$ satisfies the bound
\eqref{Ubnd}
on $(0,T_0]\times \Tbb$, and the component $\xi$ of $U$ satisfies $\xi>0$ in $(0,T_0]\times \Tbb$. Then
labeling the components of $\Pbb^\perp U(0)$ according to
\begin{equation} \label{PbbperpU(0)}
    \Pbb^\perp U(0) = (\wt_1,0,\psit,0,0)^{\tr},
\end{equation}
where $\wt_1,\psit\in H^{k-1}(\Tbb)$, there 
exist elements $\wt_2,\tilde{\eta}\in H^{k-2}(\Tbb)$ and $\nut\in H^{k-1}(\Tbb)$
such that
\begin{align}
\norm{t^{-1}w_2(t)-\ln(t)\del{\theta}\wt_{1}-\wt_2}_{H^{k-2}} &\lesssim
t+t^{2\kappa_0-3\sigma}, \label{prop:hoexp.1} \\
\norm{t^{-1}\eta-\tilde{\eta}}_{H^{k-2}} &\lesssim
t+t^{2\kappa_0-3\sigma} \label{prop:hoexp.2}
\intertext{and}
\norm{\ln(\xi(t))-((a+\wt_1)^2-1) \ln(t)-\nut}_{H^{k-1}}&\lesssim t+t^{2\kappa_0-2\sigma}, \label{prop:hoexp.3}
\end{align}
for $0<t\le T_0$.

Finally, if the Cauchy data \eqref{eq:HOCD} satisfy the constraint \eqref{eq:vacCC.1} at $t=T_0$, then the functions $\wt_1,\wt_2, \psit, \nut,\tilde{\eta}$ satisfy the following asymptotic form of the constraint equation \eqref{eq:vacCC.1}
\begin{equation}
  \label{eq:AsymptConstr}
  \partial_\theta \nut
        - 2 (a+\wt_1) \wt_2
        + \frac 12 (1+\psit)^{-1} \tilde\eta 
        =0.
\end{equation}
\end{prop}
\begin{proof}
Setting
\begin{equation} \label{Acdef}
    \Ac := \bigl((B^0)^{-1}\Bc\bigr)|_{U=0}=
    \begin{pmatrix} 
     1 & 0 & 0 & 0 & 0\\
     0 & 1 & 0 & 0 & 0\\
     0 & 0 & 1 & 0 & 0\\
     0 & 0 & 0 & a^2-1 & 0\\
     0 & 0 & 0 & 0 & 1\end{pmatrix},
\end{equation}
we can express \eqref{vacIVP.1} as
\begin{equation*}
    \del{t}U=\frac{1}{t}\Ac \Pbb U
    +\frac{1}{t}\bigl((B^0)^{-1}\Bc-\Ac\bigr)\Pbb U + \frac{1}{t}(B^0)^{-1}F-(B^0)^{-1}B^1\del{\theta}U.
\end{equation*}
Multiplying this equation on the left
by 
\begin{equation} \label{exp-mat}
    e^{-\ln(t)\Ac\Pbb} =  \begin{pmatrix} 
     1 & 0 & 0 & 0 & 0\\
     0 & \frac{1}{t} & 0 & 0 & 0\\
     0 & 0 & 1 & 0 & 0\\
     0 & 0 & 0 & \frac{1}{t^{a^2-1}} & 0\\
     0 & 0 & 0 & 0 & \frac{1}{t}\end{pmatrix},
\end{equation}
a short calculation shows that
\begin{equation} \label{Vdef}
    V= (V_1,V_2,V_3,V_4,V_5)^{\tr} :=  e^{-\ln(t)\Ac\Pbb}U
\end{equation}
satisfies
\begin{equation} \label{V-evolve-A}
     \del{t}V=
    \frac{1}{t} e^{-\ln(t)\Ac\Pbb}\bigl((B^0)^{-1}\Bc-\Ac\bigr)\Pbb U + \frac{1}{t} e^{-\ln(t)\Ac\Pbb}(B^0)^{-1}F- e^{-\ln(t)\Ac\Pbb}(B^0)^{-1}B^1\del{\theta}U.
\end{equation}
From \eqref{Udef}-\eqref{Fdef}, \eqref{PbbperpU(0)} and \eqref{exp-mat}, we then find after a straightforward calculation that we can express \eqref{V-evolve-A} as
\begin{equation} \label{V-evolve-B}
    \del{t}V = 
    \frac{1}{t}\begin{pmatrix}0 \\ \del{\theta}\wt_1\\0\\0\\0\end{pmatrix}
    + \frac{1}{t}\begin{pmatrix} 0 \\ 0\\0\\  (2a\wt_1+\wt_1^2 + f)V_4 \\ h V_2 -m(1+\psi)\xi^2 V_5
    \end{pmatrix} 
    + \frac{1}{t}\begin{pmatrix}
\frac{1}{2}\eta w_2-\frac{m}{2}(1+\psi)(a+w_1)\xi^2+t(1+\psi)\del{\theta}w_2 \\
\del{\theta}(w_1-\wt_1)\\
-m(1+\psi)^2\xi^2 \\ 0 \\ 0
\end{pmatrix}, 
\end{equation}
where
\begin{align}
f&=2aw_1+w_1^2-(2a \wt_1+\wt_1^2)+\frac{1}{4}\bigl(4(1+\psi)w_2^2+m(1+\psi)\xi^2\bigr) \label{fdef}
\intertext{and}
h&= -4m(a+w_1)(1+\psi)^2\xi^2. \label{hdef}
\end{align}

The form of \eqref{V-evolve-B} motivates us to introduce the new variable
\begin{equation} \label{Wdef}
    W:= \begin{pmatrix} V_1 \\ V_2-\ln(t)\del{\theta}\wt_1 \\V_3 \\ t^{-(2a \wt_1+\wt_1^2)}V_4\\ V_5
    \end{pmatrix}
    = \begin{pmatrix} w_1 \\ t^{-1}w_2-\ln(t)\del{\theta}\wt_1 \\\psi \\ t^{-(a^2-1 +2a \wt_1+\wt_1^2)}\xi\\ t^{-1}\eta
    \end{pmatrix},
\end{equation}
which then allows us to express \eqref{V-evolve-B} as
\begin{equation} \label{W-evolve}
    \del{t}W=\Cc W + \Fc
\end{equation}
where
\begin{align}
    \Cc &= \begin{pmatrix} 
     0 & 0 & 0 & 0 & 0\\
     0 & 0 & 0 & 0 & 0\\
     0 & 0 & 0 & 0 & 0\\
     0 & 0 & 0 & \frac{1}{t}f & 0\\
     0 & \frac{1}{t} h & 0 & 0 & -\frac{m}{t}(1+\psi)\xi^2 \end{pmatrix} \label{Ccdef}
    \intertext{and}
    \Fc &= \frac{1}{t}\begin{pmatrix}
\frac{1}{2}\eta w_2-\frac{m}{2}(1+\psi)(a+w_1)\xi^2+t(1+\psi)\del{\theta}w_2 \\
\del{\theta}(w_1-\wt_1)\\
-m(1+\psi)^2\xi^2 \\ 0 \\ \ln(t)h\del{\theta}\wt_1
\end{pmatrix}, \label{Fcdef}
\end{align}
both depending on $t$, and $W$ and $\partial_\theta W$. 

Integrating \eqref{W-evolve} in time yields
\begin{equation} \label{W-int}
    W(t)=W(t_0) + \int_{t_0}^t \Cc(\tau, W(\tau))W(\tau)+\Fc(\tau,W(\tau),\partial_\theta W(\tau))\,d\tau
\end{equation}
for $0<t_0\le t\le T_0$.
By the triangle inequality, and the Sobolev and product estimates, see Proposition 2.4 and 3.7 from Chapter 13 of  \cite{TaylorIII:1996},
we find, since $k-2\geq 1>1/2$, that
\begin{equation*}
    \norm{W(t)}_{H^{k-2}} \leq \norm{W(T_0)}_{H^{k-2}}+\int^{T_0}_t 
    \norm{\Cc(\tau, W(\tau))}_{H^{k-2}}
    \norm{W(\tau)}_{H^{k-2}}+\norm{\Fc(\tau, W(\tau),\partial_\theta W(\tau))}_{H^{k-2}}\,d\tau.
\end{equation*}
From this we conclude via an application of Gr\"onwall's inequality that
\begin{equation} \label{W-bnd-A}
\norm{W(t)}_{H^{k-2}} \leq e^{\int^{T_0}_t 
    \norm{\Cc(\tau, W(\tau))}_{H^{k-2}}\,d\tau}
\biggl(\norm{W(T_0)}_{H^{k-2}}+\int^{T_0}_t\norm{\Fc(\tau, W(\tau),\partial_\theta W(\tau))}_{H^{k-2}}\,d\tau\biggr).
\end{equation}
But by \eqref{PbbperpU(0)}, \eqref{fdef}, \eqref{hdef}, \eqref{Ccdef} and \eqref{Fcdef}, we observe, with the help of
the energy and decay estimates \eqref{eq:energyestimates}-\eqref{eq:decayestimates} and the Sobolev and product estimates, see Proposition 2.4 and 3.7 from Chapter 13 of  \cite{TaylorIII:1996}, that
\begin{equation} \label{W-bnd-B}
    \int_{t_0}^t\norm{\Cc(\tau, W(\tau))}_{H^{k-2}}+ \norm{\Fc(\tau, W(\tau),\partial_\theta W(\tau))}_{H^{k-2}}\,d\tau \lesssim \bigl(t+t^{2\kappa_0-3\sigma}\bigr)-\bigl(t_0+t_0^{2\kappa_0-3\sigma}\bigr).
\end{equation}
Thus, by \eqref{W-bnd-A}, we have 
\begin{equation*}
   \sup_{0<t<T_0} \norm{W(t)}_{H^{k-2}} \lesssim 1.
\end{equation*}
With the help of this uniform bound, we deduce
from \eqref{W-int} and another application of the product, Sobolev, and triangle inequalities, that
\begin{align*}
\norm{W(t)-W(t_0)}_{H^{k-2}} 
&\leq   \int_{t_0}^t\norm{\Cc(\tau, W(\tau))}_{H^{k-2}}
\norm{W(\tau)}_{H^{k-2}}\,d\tau + \int_{t_0}^t\norm{\Fc(\tau, W(\tau),\partial_\theta W(\tau))}_{H^{k-2}}
\,d\tau \\
&\lesssim  \int_{t_0}^t\norm{\Cc(\tau, W(\tau))}_{H^{k-2}}+ \norm{\Fc(\tau, W(\tau),\partial_\theta W(\tau))}_{H^{k-2}}\,d\tau.
\end{align*}
From this inequality and \eqref{W-bnd-B}, we deduce that the limit $\lim_{t\searrow 0} W(t)$ converges to an element of $H^{k-2}(\Tbb)$, and denoting this element by $W(0)$, we can extend $W(t)$ to a uniformly continuous map on $[0,T_0]$; that is
\begin{equation*}
    W\in C^0\bigl([0,T_0],H^{k-2}(\Tbb,\Rbb^5)\bigr),
\end{equation*}
and moreover, that
\begin{equation} \label{W-bnd-C}
\norm{W(t)-W(0)}_{H^{k-2}} 
\lesssim  t^{2\kappa_0-3\sigma}+t
\end{equation}
for $0<t\le T_0$. The stated estimates \eqref{prop:hoexp.1} and \eqref{prop:hoexp.2} are then a direct consequence of \eqref{Wdef} and \eqref{W-bnd-C}.

We observe that the previous arguments also yield an estimate for $\xi$ (in addition to the estimates \eqref{prop:hoexp.1} and \eqref{prop:hoexp.2} for $w_2$ and $\eta$, respectively) via the estimate for the component $W_4$ resulting from \eqref{W-bnd-C}.
However, since the equation
\begin{equation}
  \partial_t W_4=\frac 1t f W_4
\end{equation}
decouples from the rest of the system \eqref{W-evolve}, we can establish the estimate \eqref{prop:hoexp.2}, with improved regularity.  
Note that $\xi$, and therefore $W_4$, is strictly positive everywhere on $(0,T_0]\times\Tbb^3$ as a consequence of Proposition~\ref{prop:main_existence} and \eqref{Wdef}.
It follows that
\begin{equation}
  \ln(W_4(t))-\ln(W_4(t_0))=\int_{t_0}^t f(\tau)\tau^{-1}d\tau.
\end{equation}
From the definition of $f$, \eqref{fdef}, and Proposition~\ref{prop:main_existence}, we obtain $\norm{f(\tau)}_{H^{k-1}}\lesssim t^{2\kappa_0-2\sigma}+t$, and hence
\begin{equation}
  \norm{\ln(W_4(t))-\ln(W_4(t_0))}_{H^{k-1}}\lesssim t^{2\kappa_0-2\sigma}+t-(t_0^{2\kappa_0-2\sigma}+t_0)
\end{equation}
for any $0<t_0\le t\le T_0$. Thanks to the completeness of $H^{k-1}(\Tbb)$, the sequence $\ln(W_4(t))$ converges in the limit $t\searrow 0$ with respect to the $H^{k-1}$-norm and has a limit, which we call $\nut\in H^{k-1}(\Tbb)$, such that
\begin{equation}
  \label{eq:fljsdlfgjlsdfnvxc}
  \norm{\ln(W_4(t))-\nut}_{H^{k-1}}\lesssim t^{2\kappa_0-2\sigma}+t,
\end{equation}
for all $t\in (0,T_0]$. The estimate \eqref{prop:hoexp.2} follows from \eqref{eq:fljsdlfgjlsdfnvxc} together with the identity
\[\ln(W_4(t))=\ln(\xi(t))-((a+\wt_1)^2-1) \ln(t)\]
obtained from \eqref{Wdef}.

To finish the proof, suppose that the Cauchy data \eqref{eq:HOCD} satisfy the constraint \eqref{eq:vacCC.1} at $t=T_0$. 
We note that this constraint propagates as a consequence of the propagation of \eqref{eq:vacC.1} and the algebraically defined variables \eqref{nvars.1}-\eqref{nvars.6}.
Hence, \eqref{eq:vacCC.1} holds for all $t\in (0,T_0]$.
For any such $t$ we multiply this equation \eqref{eq:vacCC.1} by $t^{-(a+\wt_1)^2}$ to find that
\begin{align*}
  0=
    \bnorm{&
      t^{-((a+\wt_1)^2-1)}\del{\theta}\xi         
      - 2 (a+w_1) t^{-1} w_2 t^{-((a+\wt_1)^2-1)}\xi
      + \frac 12 (1+\psi)^{-1} t^{-1}\eta t^{-((a+\wt_1)^2-1)}\xi
    }_{H^{k-3}}\\
  =
    \bnorm{&
      2 \ln t\partial_\theta \wt_1 (a+\wt_1) t^{-((a+\wt_1)^2-1)}\xi
      +\partial_\theta \bigl(t^{-((a+\wt_1)^2-1)}\xi\bigr)\\
      &- 2 (a+w_1) t^{-1} w_2 t^{-((a+\wt_1)^2-1)}\xi\\
      &+ \frac 12 (1+\psi)^{-1} t^{-1}\eta t^{-((a+\wt_1)^2-1)}\xi
    }_{H^{k-3}} \\
  \ge 
    \bnorm{&
      2 \ln t\partial_\theta \wt_1 (a+\wt_1) e^{\nut}
      +\partial_\theta e^{\nut}
      - 2 (a+\wt_1) (\ln t\partial_\theta\wt_1  +\wt_2)e^{\nut}
      + \frac 12 (1+\psit)^{-1} \tilde\eta e^{\nut}  
    }_{H^{k-3}}\\ 
    -\bnorm{&
      2 \ln t\partial_\theta \wt_1 (a+\wt_1) (t^{-((a+\wt_1)^2-1)}\xi-e^{\nut})
    }_{H^{k-2}}
    -\bnorm{
      t^{-((a+\wt_1)^2-1)}\xi-e^{\nut}
    }_{H^{k-2}}\\
    -\bnorm{&
      - 2 (w_1-\wt_1) t^{-1} w_2 t^{-((a+\wt_1)^2-1)}\xi
    }_{H^{k-2}}\\
    -\bnorm{&
      - 2 (a+\wt_1) (t^{-1} w_2-\ln t\partial_\theta\wt_1-\wt_2) t^{-((a+\wt_1)^2-1)}\xi
    }_{H^{k-2}}\\
    -\bnorm{&
      - 2 (a+\wt_1) (\ln t\partial_\theta\wt_1+\wt_2) (t^{-((a+\wt_1)^2-1)}\xi-e^{\nut})
    }_{H^{k-2}}\\
    -\bnorm{&
      \frac 12 \frac{\psi-\psit}{(1+\psi)(1+\psit)} t^{-1}\eta t^{-((a+\wt_1)^2-1)}\xi
    }_{H^{k-2}}\\
    -\bnorm{&
      \frac 12 (1+\psit)^{-1} (t^{-1}\eta-\tilde\eta) t^{-((a+\wt_1)^2-1)}\xi
    }_{H^{k-2}}
    -\bnorm{
      \frac 12 (1+\psit)^{-1} t^{-1}\tilde\eta (t^{-((a+\wt_1)^2-1)}\xi-e^{\nut})
    }_{H^{k-2}},
\end{align*}
where in deriving this inequality we have used the fact that $k-2>1/2$.
Since $|\psi|$ and therefore $|\psit|$ are both strictly smaller than $1$ as a consequence of \eqref{Ubnd} with \eqref{psibndF},
and $(t^{-((a+\wt_1)^2-1)}\xi-e^{\nut}) = e^{\nut}(\exp(\ln(\xi)-((a+\wt_1)^2-1)\ln(t)-\nut)-1)$, we conclude from \eqref{eq:decayestimates}, \eqref{prop:hoexp.1}-\eqref{prop:hoexp.3} and the Moser inequality that all but the first line of this previous estimate go to zero in the limit $t\searrow 0$. 
This therefore implies the validity of the asymptotic constraint \eqref{eq:AsymptConstr}.
\end{proof}

\subsection{Solutions to the Full Polarized $\Tbb^2$-symmetric Vacuum Einstein Equations for Perturbations of Kasner Initial Data}

Propositions~\ref{prop:main_existence} and \ref{prop:hoexp} establish the existence of global solutions of the initial value problem \eqref{vacIVP.1}-\eqref{vacIVP.2} of the \emph{core evolution system} \eqref{vacB.2}-\eqref{vacB.6}, the leading-order behavior of the corresponding variables $\bigl(w_1,w_2,\psi,\xi,\eta\bigr)^{\tr}$, and the existence of limits $\wt_1, \psit, \nut \in H^{k-1}(\Tbb)$ and $\wt_2,\tilde{\eta}\in H^{k-2}(\Tbb)$. 
We have also addressed the constraint \eqref{eq:vacCC.1}, which, if it is satisfied by the initial data, propagates and then implies \eqref{eq:AsymptConstr}. 
In order to construct solutions to the \emph{full} vacuum Einstein equations for polarized $T^2$-symmetric spacetimes, it remains to solve the decoupled evolution equation \eqref{vacB.1} for $w_0$, \eqref{Evac.G.t} for $G$ and \eqref{Evac.H.t} for $H$ together with the constraints \eqref{eq:vacCC.2} and \eqref{eq:vacCC.3}.

\begin{prop}
  \label{prop:fullEFE}
  Consider the same conditions as specified in the hypothesis for Proposition~\ref{prop:hoexp}. 
  Let $U$ be the solution to the GIVP \eqref{vacIVP.1}-\eqref{vacIVP.2} determined by Cauchy data $\mathring{U}$ with $\norm{\mathring{U}}_{H^k}< R_0$. 
  Given any $\mathring{w}_0, \mathring{G}, \mathring{H}\in H^{k}(\Tbb)$, then the Cauchy problems of \eqref{vacB.1} with Cauchy data $\mathring{w}_0$, of \eqref{Evac.G.t} with Cauchy data $\mathring{G}$, and of \eqref{Evac.H.t} with Cauchy data $\mathring{H}$  imposed at $t=T_0$ have unique solutions
  \[w_0, G, H\in C^1\bigl((0,T_0],H^k(\Tbb)\bigr)\cap L^\infty\bigl((0,T_0],H^k(\Tbb)\bigr),\]
  where $G=\mathring{G}$. 
  Provided $R_0$ is sufficiently small, there exists $\wt_0\in H^{k-1}(\Tbb)$ and $\Ht\in H^{k-1}(\Tbb)$ such that
  \begin{align}
    \label{eq:OtherLimit.1}
    \Bnorm{w_0(t)-\frac1 a \Bigl(\wt_1+a\Bigr)\ln t-\frac1 a \wt_0}_{H^{k-1}}&\lesssim t+t^{2\kappa_0-2\sigma},\\
    \label{eq:OtherLimit.2}
    \Bnorm{H(t)-\Ht}_{H^{k-1}}&\lesssim t^{2(\gamma_{\min}-1-\sigma)},
  \end{align}
  where $\gamma_{\min}=\min_{\theta\in [0,2\pi)}\{(a+\wt_1(\theta))^2\}$.
  Finally, if the Cauchy data satisfy the constraints \eqref{eq:vacCC.2} and \eqref{eq:vacCC.3} at $t=T_0$, then
\begin{align}
  \label{eq:AsymptConstr2}
  \partial_\theta\wt_0=\wt_2, \\
  \label{eq:AsymptConstr3}
  \del{\theta}\psit=\etat.
\end{align}
\end{prop}

\begin{proof}
  The existence and regularity of the solution $w_0$ follows directly from \eqref{vacB.1} and results concerning $w_1$ in Proposition~\ref{prop:main_existence}.
  Given that \eqref{vacB.1} also implies
  \[\del{t}\Bigl(w_0-\frac1 a \Bigl(\wt_1+a\Bigr)\ln t\Bigr)
    =\frac{1}{a t}(w_1-\wt_1),\]
  it follows that
  \begin{equation}
    \label{eq:yeq}
    y(t):=w_0(t)-\frac{1}{a}\Bigl(\wt_1+a\Bigr)\ln t
    =y_0+\int_{T_0}^t\frac{1}{a}(w_1(s)-\wt_1) s^{-1}ds
  \end{equation}
  for some $y_0\in H^{k-1}(\Tbb)$. 
  For any monotonic sequence $(t_n)$ approaching zero, we therefore get
  \begin{align}
    \norm{y(t_n)-y(t_m)}_{H^{k-1}}
    &\le \left|\frac{1}{a}\right| \left|\int_{t_m}^{t_n}\norm{w_1(s)-\wt_1}_{H^{k-1}} s^{-1}ds\right|\notag\\
    \label{eq:ybd}
    &\lesssim \left|\int_{t_m}^{t_n} (s+s^{2\kappa_0-2\sigma}) s^{-1}ds\right|
    \lesssim |t_n-t_m|+|t_n^{2\kappa_0-2\sigma}-t_m^{2\kappa_0-2\sigma}|,
  \end{align}
  using \eqref{eq:decayestimates}. 
  The sequence $(y(t_n))$ is therefore a Cauchy sequence in $H^{k-1}(\Tbb)$ and hence converges to a limit, $\frac{1}{a} \wt_0\in H^{k-1}(\Tbb)$.
  We choose the rescaling by ${a}$ for later convenience.
  Setting $t_n=t$ for an arbitrary $t\in (0,T_0]$ we can take the limit $t_m\searrow 0$ to obtain \eqref{eq:OtherLimit.1}.

  The evolution equation \eqref{Evac.G.t} for $G$ is trivial.
  The evolution equation \eqref{Evac.H.t} for $H$ reads in our variables
  \[\del{t}H = \frac 1t \sqrt{m} \sqrt{1+\psi}\, t^{2((a+\wt_1)^2-1)}e^{2\nut}\exp\bigl(2(\ln(\xi)-((a+\wt_1)^2-1) \ln(t)-\nut)\bigr),\]
  and therefore
  \[H(t)-H(t_0)=\sqrt{m}\,e^{2\nut}\int_{t_0}^t \sqrt{1+\psi(\tau)}\exp\bigl(2(\ln(\xi)-((a+\wt_1)^2-1) \ln(t)-\nut)\bigr) \tau^{2((a+\wt_1)^2-1)-1} d\tau\]
  for any $t,t_0\in (0,T_0]$.
  Since $a^2>1$, we conclude from Propositions~\ref{prop:main_existence} and \ref{prop:hoexp} as well as the Moser inequality that
\begin{align*}
    \norm{ H(t)-H(t_0)}_{H^{k-1}}
  &\lesssim \int_{t_0}^t \norm{\tau^{2((a+\wt_1)^2-\gamma_{\min})}}_{H^{k-1}} \tau^{2(\gamma_{\min}-1)-1}d\tau\\
  &\lesssim t^{2(\gamma_{\min}-1-\sigma)}-t_0^{2(\gamma_{\min}-1-\sigma)}
  \end{align*}
  so long as $0<t_0\le t\le T_0$, since $\norm{\tau^{2((a+\wt_1)^2-\gamma_{\min})}}_{H^{k-1}}$ is uniformly bounded on $[0,T_0]$
and since it follows from  Proposition~\ref{prop:main_existence} that $\wt_1$ can be made so small
   that $\gamma_{\min}>1$. Hence $H(t)$ converges with respect to the $H^{k-1}$-norm in the limit $t\searrow 0$
  owing to the completeness of $H^{k-1}(\Tbb)$. As a consequence, the limit $\Ht$ lies in $H^{k-1}(\Tbb)$ and \eqref{eq:OtherLimit.2} follows.

  If the constraint \eqref{eq:vacCC.2} is satisfied at $t=T_0$, then  \eqref{eq:vacCC.2} is satisfied at every $t\in (0,T_0]$ and
  \begin{align*}
    0=&\bnorm{a\del{\theta}w_0-t^{-1}w_2}_{H^{k-2}}\\
    \ge &\Bnorm{{a}\del{\theta}\Bigl(\frac{1}{a} \Bigl(\wt_1+a\Bigr)\ln t+\frac{1}{a}\wt_0\Bigr)-(\ln(t)\del{\theta}\wt_{1}+\wt_2)}_{H^{k-2}}\\
    &-\left|{a}\right|\Bnorm{w_0-\frac{1}{a}\Bigl(\Bigl(\wt_1+a\Bigr)\ln t+\wt_0\Bigr)}_{H^{k-1}}
    -\bnorm{t^{-1}w_2 -(\ln(t)\del{\theta}\wt_{1}+\wt_2)}_{H^{k-2}}.
  \end{align*}
  According to \eqref{prop:hoexp.1} and \eqref{eq:OtherLimit.1}, we can make the last two terms arbitrarily small in the limit $t\searrow 0$. 
  Similarly, 
  \begin{align*}
    0 =& \bnorm{t^{-1}\eta - \del{\theta}\psi}_{H^{k-2}} \\
      \ge& \bnorm{\etat - \del{\theta}\psit}_{H^{k-2}}
            - \bnorm{t^{-1}\eta - \etat}_{H^{k-2}} 
            - \Bnorm{\psi - \psit}_{H^{k-1}},
  \end{align*}
where it follows from \eqref{eq:decayestimates} and \eqref{prop:hoexp.2} that the latter two terms vanish in the limit $t\searrow 0$.
This implies \eqref{eq:AsymptConstr2} and \eqref{eq:AsymptConstr3}.
\end{proof}

Before concluding this subsection, we establish further estimates on the time-derivatives of $e^\nu = t \xi$ and $H$. 
These estimates will be employed in Section~\ref{s.proof_main_result} to obtain decay estimates for the differences between time derivatives of solutions of the full Einstein vacuum system and of the VTD system. 
\begin{lem}
  \label{lem.time_derivatives}
    Suppose the conditions of Proposition~\ref{prop:hoexp} and \ref{prop:fullEFE} are satisfied, and let $w_0, w_1, w_2, \psi, \eta, \xi, H, G$ be the corresponding solution of \eqref{vacB.1}-\eqref{vacB.6}, \eqref{eq:vacCC.2}-\eqref{eq:vacCC.3}, and \eqref{Evac.H.t}-\eqref{Evac.G.t}.
    Then the estimates
    \begin{align}
      \label{e.asymptotics.dnudt}
      \Bnorm{t\del{t}\nu - (a + \wt_1)^2}_{H^{k-1}} 
        &\lesssim t^{1-\sigma} + t^{2\kappa_0 - 2\sigma} + t^{2(\gamma_{\min}-1-\sigma)}, \\
      \label{e.asymptotics.dHdt}
      \Bnorm{t\del{t}H - \sqrt{m}(1+\psit)^{1/2}t^{2(a+\wt_1)^2-2}e^{2\nut}}_{H^{k-1}} &\lesssim t^{2\gamma_{\min}-2-\sigma}(t + t^{2\kappa_0 - 3\sigma}), \\
      \label{e.asymptotics.dadt}
      \Bnorm{t\del{t}\psi + \sqrt{m}(1+\psit)^2t^{2(a+\wt_1)^2-2}e^{2\nut}}_{H^{k-1}} &\lesssim t^{2\gamma_{\min}-2-\sigma}(t + t^{2\kappa_0 - 3\sigma}),
    \end{align}
    hold for $\sigma>0$ that can be chosen arbitrarily small.
\end{lem}
The lemma is established by first using \eqref{eq:decayestimates}, \eqref{prop:hoexp.1}-\eqref{prop:hoexp.3}, and \eqref{eq:OtherLimit.1}-\eqref{eq:OtherLimit.2} to estimate the leading order asymptotic behavior of the right-hand-sides of \eqref{Evac.H.t}, \eqref{vacB.5}, and \eqref{vacB.5} and then applying the product and Moser estimates along with Sobolev's inequality.

\subsection{Stability of Kasner Solutions within the Vacuum Polarized $\Tbb^2$-symmetric Class}
\label{s.stability_of_kasner}

In this subsection, we show that there exists a sub-family of Kasner solutions that fall into the set of solutions whose existence is guaranteed by Proposition~\ref{prop:main_existence} and \ref{prop:fullEFE}. 
The other (non-Kasner) solutions guaranteed by Proposition~\ref{prop:main_existence} with the same choice of the parameter $a$ are then interpreted as nonlinear perturbations of the Kasner sub-family.
The $H^k$ estimates proved below strengthen the $L^\infty$-bound \eqref{Ubnd} guaranteed by the global existence result Proposition~\ref{prop:main_existence}.
This result provides a notion of stability for the Kasner solutions within the polarized $\Tbb^2$-symmetric vacuum class. 

We note the variables used in Proposition~\ref{prop:main_existence} can be expressed, due to \eqref{first_order_variables} and \eqref{nvars.1}-\eqref{nvars.6}, in terms of the metric fields of \eqref{polT2metric} as
\begin{equation}
\label{T2vars2FOvars}
    \xi = t^{-1} e^\nu,\quad
    \psi= \alpha-1,\quad
    \eta= t \partial_\theta \alpha, \quad
    w_0=\frac{1}{a} u,\quad
    w_1 = t \partial_t u - a,\quad
    w_2 = t \partial_\theta u.
\end{equation}
Using \eqref{T2vars2FOvars}, we see that the Kasner solutions \eqref{kasner.fields} therefore correspond to 
  \begin{equation}
    \label{Kasner2FOvars}
    \begin{split}
    \xi^{(K)} = t^{\frac 14(K-3)(K+1)},\quad
    \psi^{(K)}= 0,\quad
    w_0^{(K)}=\frac{1}{2a}(1-K)\ln t,\quad
    w_1^{(K)} = \frac 12(1-K) - a,\\
    \eta^{(K)}=w_2^{(K)} = G^{(K)}=H^{(K)}=0.
  \end{split}
\end{equation}
  It is easy to confirm that for any $K\in\Rbb$, this is a solution of the full Einstein vacuum equations in the form \eqref{vacB.1}-\eqref{eq:vacCC.3} for $m=0$. Given an arbitrary $T_0>0$, it is in fact the solution of the initial value problem of these equations for Cauchy data
  \begin{equation}
    \label{eq:KasnerCD}
    \mathring{U}^{(K)} = \Bigl(\frac 12 (1-K) -a, 0, 0, T_0^{\frac 14(K-3)(K+1)}, 0\Bigr)^{\tr},\quad
    \mathring{w}_0^{(K)}=\frac{1}{2a}(1-K)\ln T_0,\quad
    \mathring{G}^{(K)}=\mathring{H}^{(K)}=0,
  \end{equation}
  imposed at $t=T_0$;
  cf.\ \eqref{eq:HOCD}. In particular it follows that the Kasner solution corresponding to an arbitrary $K$ with $K>3$ or $K<-1$, i.e.,
\begin{equation}
    |K-1|>2,
\end{equation}
agrees with the
solution $(U,w_0,G,H)$ of the full Einstein vacuum equations \eqref{vacB.1} -- \eqref{eq:vacCC.2} for $m=0$ asserted by Propositions~\ref{prop:main_existence} and \ref{prop:fullEFE} for the Cauchy data \eqref{eq:KasnerCD} provided we choose $T_0$ sufficiently small and
\begin{equation}
  \label{eq:KasnerStabChoiceba}
  a=\frac 12 (1-K),
\end{equation}
which allows us to satisfy \eqref{eq:CDSmallness} because it follows that $T_0^{\frac 14(K-3)(K+1)}$ can be made as small as necessary to apply Proposition~\ref{prop:main_existence}. 
Observe however that these conclusions would be invalid if $-1\le K\le 3$, which is consistent with our discussion in Section \ref{sec:spatially_homogeneous}.
Recall that there is no smallness condition for $\mathring{w}_0$ in Propositions~\ref{prop:main_existence} and \ref{prop:fullEFE}.
It immediately follows from \eqref{eq:KasnerStabChoiceba} that $|a|>1$. 

Given an arbitrary $K$ as above,   
we set $a$ according to \eqref{eq:KasnerStabChoiceba}. We then choose $m\ge 0$ and a sufficiently small $T_0$. 
The solution $(U,w_0,G,H)$ of the full Einstein vacuum equations \eqref{vacB.1}-\eqref{eq:vacCC.2} asserted by Propositions~\ref{prop:main_existence} and \ref{prop:fullEFE} for arbitrary Cauchy data $\mathring U$ as in Proposition~\ref{prop:main_existence} and $(\mathring{w}_0, \mathring{G}, \mathring{H})$ as in Proposition~\ref{prop:fullEFE}, which satisfy the constraints \eqref{eq:vacCC.2}-\eqref{eq:vacCC.3} at $t=T_0$,
can be understood as a \emph{nonlinear perturbation of the Kasner solution} given by $K$, within the class of polarized vacuum $\Tbb^2$-symmetric solutions. 
In particular, we have from \eqref{eq:decayestimates}, \eqref{prop:hoexp.1}-\eqref{prop:hoexp.3}, \eqref{eq:OtherLimit.1}, \eqref{eq:OtherLimit.2} and \eqref{Kasner2FOvars} that
\begin{align*}
  \Bnorm{w_0-w_0^{(K)}-\frac{1}{a} \Bigl(\wt_1\ln t+\wt_0\Bigr)}_{H^{k-1}}&\lesssim t+t^{2\kappa_0-2\sigma},\\
  \Bnorm{w_1-w_1^{(K)}-\wt_1}_{H^{k-1}}&\lesssim t+t^{2\kappa_0-2\sigma},\\
  \norm{t^{-1}w_2-t^{-1}w_2^{(K)}-(\ln(t)\del{\theta}\wt_{1}+\wt_2)}_{H^{k-2}} &\lesssim
  t+t^{2\kappa_0-3\sigma},\\
  \norm{\ln(\xi(t))-\ln(\xi^{(K)}(t))-(2a\wt_1+\wt_1^2) \ln(t)-\nut}_{H^{k-1}}&\lesssim t+t^{2\kappa_0-2\sigma},\\
  \Bnorm{\psi-\psi^{(K)}-\tilde\psi}_{H^{k-1}}&\lesssim t+t^{2\kappa_0-2\sigma},\\
  \norm{t^{-1}\eta-t^{-1}\eta^{(K)}-\tilde{\eta}}_{H^{k-2}} &\lesssim t+t^{2\kappa_0-3\sigma},\\
  \Bnorm{G-G^{(K)}-\mathring G}_{H^{k}}&=0,\\
  \Bnorm{H-H^{(K)}-\Ht}_{H^{k-1}}&\lesssim t^{2(\gamma_{\min}-1-\sigma)},
\end{align*}
for all $t\in (0,T_0]$.
Using \eqref{eq:decayestimates}, \eqref{prop:hoexp.1}-\eqref{prop:hoexp.3}, \eqref{eq:OtherLimit.1} and \eqref{eq:OtherLimit.2} with $t=T_0$ to estimate the limit quantities in terms of the initial data and $T_0$, we get
\begin{align*}
  \Bnorm{w_0-w_0^{(K)}}_{H^{k-1}}&\lesssim \Bnorm{ \mathring{w}_1 }_{H^{k-1}}|\ln (t)|+\Bnorm{\mathring{w}_0-\ln T_0}_{H^{k-1}}+(1+|\ln (t)|)(T_0+T_0^{2\kappa_0-2\sigma}),\\
  \Bnorm{w_1-w_1^{(K)}}_{H^{k-1}}&\lesssim \Bnorm{ \mathring{w}_1 }_{H^{k-1}}+ T_0+T_0^{2\kappa_0-2\sigma},\\
  \norm{t^{-1}w_2-t^{-1}w_2^{(K)}}_{H^{k-2}} &\lesssim \Bnorm{ \mathring{w}_1 }_{H^{k-1}}|\ln (t)|+
  \norm{T_0^{-1}\mathring{w}_2}_{H^{k-2}}+(1+|\ln (t)|)(T_0+T_0^{2\kappa_0-3\sigma}),\\  
  \norm{\ln(\xi(t))-\ln(\xi^{(K)}(t))}_{H^{k-1}}&\lesssim \Bnorm{ \mathring{w}_1 }_{H^{k-1}}|\ln(t)|+\Bnorm{ \mathring{w}_1 }_{H^{k-1}}^2 |\ln(t)|
                                                  +\norm{\ln(\mathring{\xi})+\ln(T_0)}_{H^{k-1}}\\
                                                  &+(1+|\ln (t)|)|T_0+T_0^{2\kappa_0-2\sigma}),\\
  \Bnorm{\psi-\psi^{(K)}}_{H^{k-1}}&\lesssim \Bnorm{\mathring{\psi}}_{H^{k-1}}+T_0+T_0^{2\kappa_0-2\sigma},\\
  \norm{t^{-1}(\eta-\eta^{(K)})}_{H^{k-2}} &\lesssim \norm{T_0^{-1}\mathring{\eta}}_{H^{k-2}}+T_0+T_0^{2\kappa_0-3\sigma},\\                                                         
  \Bnorm{G-G^{(K)}}_{H^{k}}&=\Bnorm{\mathring G}_{H^{k}},\\
  \Bnorm{H-H^{(K)}}_{H^{k-1}}&\lesssim \Bnorm{\mathring{H}}_{H^{k-1}}+T_0^{2(\gamma_{\min}-1-\sigma)},
\end{align*}
assuming that $|\ln(t)|\ge|\ln(T_0)|$.
We can then divide the first, the third and the fourth inequality by $|\ln (t)|$ (assuming that $T_0$, and therefore $t$, is so small that $|\ln t|\ge |\ln(T_0)|>1$) to find
\begin{align}
  \label{s.stability_of_kasner.first}
  \Bnorm{(w_0-w_0^{(K)})/\ln (t)}_{H^{k-1}}
  &\lesssim \frac 1{|\ln T_0|}\bnorm{\mathring{w}_0-\ln T_0}_{H^{k-1}}+\bnorm{\mathring{w}_1}_{H^{k-1}}\\
  &\qquad+T_0+T_0^{2\kappa_0-2\sigma},\\
  \Bnorm{w_1-w_1^{(K)}}_{H^{k-1}}&\lesssim \bnorm{\mathring{w}_1}_{H^{k-1}}+T_0+T_0^{2\kappa_0-2\sigma},\\
  \norm{t^{-1}(w_2-w_2^{(K)})/\ln (t)}_{H^{k-2}}  
  &\lesssim \norm{T_0^{-1}\mathring{w_2}/\ln T_0}_{H^{k-2}}
    +\norm{\mathring{w}_{1}}_{H^{k-1}}
    +T_0+T_0^{2\kappa_0-3\sigma},\\
  \norm{(\ln(\xi(t))-\ln(\xi^{(K)}(t)))/\ln(t)}_{H^{k-1}}&\lesssim \Bnorm{ \mathring{w}_1 }_{H^{k-1}}+\Bnorm{ \mathring{w}_1 }_{H^{k-1}}^2 
                                                  +\frac 1{|\ln T_0|}\norm{\ln(\mathring{\xi})+\ln(T_0)}_{H^{k-1}}\\
                                                  &+T_0+T_0^{2\kappa_0-2\sigma},\\                                                                                  
  \Bnorm{\psi-\psi^{(K)}}_{H^{k-1}}&\lesssim \Bnorm{\mathring{\psi}}_{H^{k-1}}+T_0+T_0^{2\kappa_0-2\sigma},\\
  \norm{t^{-1}(\eta-\eta^{(K)})}_{H^{k-2}} &\lesssim \norm{T_0^{-1}\mathring{\eta}}_{H^{k-2}}+T_0+T_0^{2\kappa_0-3\sigma},\\                                                         
  \Bnorm{G-G^{(K)}}_{H^{k}}&=\Bnorm{\mathring G}_{H^{k}},\\
  \label{s.stability_of_kasner.last}
  \Bnorm{H-H^{(K)}}_{H^{k-1}}&\lesssim \Bnorm{\mathring{H}}_{H^{k-1}}+T_0^{2(\gamma_{\min}-1-\sigma)},
\end{align}
where we recall that the exact Kasner quantities are given by \eqref{Kasner2FOvars}. 
The above inequalities bound the difference between (1) an arbitrary solution of the full (polarized $\Tbb^2$-symmetric) Einstein vacuum equations \eqref{vacB.1}-\eqref{eq:vacCC.2} as asserted by Propositions~\ref{prop:main_existence} and \ref{prop:fullEFE} for arbitrary Cauchy data $\mathring U$ as in Proposition~\ref{prop:main_existence} and $(\mathring{w}_0, \mathring{G}, \mathring{H})$ as in Proposition~\ref{prop:fullEFE} which satisfy the constraints \eqref{eq:vacCC.1}--\eqref{eq:vacCC.3} at $t=T_0$, and, (2) the Kasner solution given by $K$ uniformly for $t\in (0,T_0]$. 
These bounds are given in terms of the size of the initial data and the size of $T_0$. 
In particular, this difference is therefore uniformly small provided $T_0$ is small and 
\begin{gather*}
  \frac 1{|\ln T_0|}\Bnorm{\mathring{w}_0-\ln T_0}_{H^{k-1}},\quad
  \left|\right|\bnorm{\mathring{w}_1}_{H^{k-1}}, \\
  \norm{T_0^{-1}\mathring{w_2}/\ln T_0}_{H^{k-2}},\quad \frac 1{|\ln T_0|}\norm{\ln(\mathring{\xi})+\ln(T_0)}_{H^{k-1}},\quad \Bnorm{\mathring{\psi}}_{H^{k-1}}, \\
  \norm{T_0^{-1}\mathring{\eta}}_{H^{k-2}},\quad
  \Bnorm{\mathring G}_{H^{k}},
  \Bnorm{\mathring{H}}_{H^{k-1}},
\end{gather*}
are small.

\subsection{Existence of Solutions to the Singular Initial Value Problem for the VTD Equations}
\label{s.avtd}

\phantom{h}
\vspace{1ex}
 
Propositions \ref{prop:hoexp}, \ref{prop:fullEFE}, and \ref{prop:main_existence} above establish the existence of solutions to the full Einstein vacuum equations of the form
\begin{align}
  w_0 &= \frac{1}{a} \left(\wt_1 + a\right) \ln(t) + \frac{1}{a} \wt_{0} + O(t + t^{2\kappa_0 - 2\sigma}) \\
  w_1 &= \wt_1 + O(t + t^{2\kappa_0 - 2\sigma}) \\ 
  w_2 &= t\wt_2 + t\ln(t)\del{\theta}\wt_1 + O(t^2 + t^{1 + 2\kappa_0 - 3\sigma}) \\ 
  \psi &= \psit +  O(t + t^{2\kappa_0 - 2\sigma}) \\
  \eta &= t \etat + O(t^2 + t^{1 + 2\kappa_0 - 3\sigma}) \\ 
  \nu &= (a + \wt_1)^2\ln(t) + \nut + O(t + t^{2\kappa_0 - 2\sigma}) \\  
  H &= \Ht + O(t^{2(\gamma_{\min}-1-\sigma)}) \\ 
  G &= \mathring{G}
\end{align}
where 
\begin{equation}
  \mathring{G} \in H^k(\Tbb), \quad \wt_0, \wt_1, \psit, \nut, \Ht \in H^{k-1}(\Tbb), \AND \wt_2, \etat \in H^{k-2}(\Tbb).
\end{equation}
We show below that such solutions are asymptotically velocity term dominated (AVTD). 
The leading order terms in the expressions above have heuristically been shown to satisfy the VTD system \cite{Clausen2007,isenberg1999}.
To make this argument rigorous, we prove existence of solutions to the VTD system \eqref{vtd.u}--\eqref{vtd.nu} with the above leading order asymptotics.
The technique is standard for Fuchsian ODE; see for example Theorem 5.1 of \cite{kichenassamy2007k}.
The system of ODE considered here is parameterized by $\theta \in \Tbb^1$, though no spatial derivatives occur in the main system of equations. 
In order to compare the solutions of the VTD system that we obtain with solutions of the full Einstein system (which does contain spatial derivative terms), below we employ estimates in a Sobolev space. 
This comparison, and thus the verification of AVTD behavior, is performed in Section~\ref{s.proof_main_result}.

It is convenient to use the variables introduced above in \eqref{nvars.1}--\eqref{nvars.6}. 
In terms of these variables, the VTD equations \eqref{vtd.u}--\eqref{vtd.nu} can be written as
\begin{align}
  \label{vtd.1}
  t \del{t}w_0 &= \frac{1}{a} w_1 + 1, \\
  \label{vtd.2}
  t \del{t}w_1 &= -\frac{m}{2}(1 + \psi)( a + w_1) \xi^2, \\
  \label{vtd.3}
  t \del{t}\psi &= - m (b + \psi)^2 \xi^2, \\
  \label{vtd.4}
  t \del{t}\xi &= \left( (a+ w_1)^2 - 1\right)\xi + \frac{m}{4}(1 + \psi)\xi^3.
\end{align}

Before stating the main result of this subsection we introduce the Banach space
\begin{equation}
  \label{bdd_decay_space}
  E_{k, \mu, T} = \{ V \in C([0,T], H^k(\Tbb, \Rbb^4)) : \sup_{t \in (0, T)}\|t^{-\mu}V\|_{H^k} < \infty \},
\end{equation}
of $H^k$-valued functions which are continuous on $[0, T]$ and vanish faster than $t^\mu, \, \mu \in \Rbb_+$ as $t\searrow 0$. 
Below, we assume $k\ge1$ to ensure continuity in space.

\begin{prop}
  \label{prop.vtd_existence}
  Let $m>0$, $k\ge1$. 
  Let $\wb_0,\wb_1, \psib,\xib, \Hb, \Gb \in H^{k}(\Tbb)$, and let $\mu \in C^\infty(\Tbb)$ satisfy $0 < \mu < 2(\gamma-1)$, where $\gamma := (a + \wb_1)^2 > 1$ for all $\theta \in \Tbb$.  
  For any such choice of $m, k, \mu, \wb_0,\psib, \wb_1, \xib, \Hb, \Gb$, there exists a solution of the VTD equations \eqref{vtd.1}--\eqref{vtd.4} of the form 
  \begin{equation}
    \label{fuchsian_ansatz}
    w_0=\frac{1}{a} \left((\wb_1 + a)\ln(t) + \wb_0 \right)+ \omega_0,\quad
    w_1=\wb_1 + \omega_1, \quad
    \psi=\psib + \omega_2, \quad
    \xi=t^{\gamma-1} (\xib + \omega_3),
  \end{equation}
  where $\Omega := \left(\omega_0, \omega_1, \omega_2, \omega_3 \right)$ is in the space $E_{k, \mu, T}$.

Further, provided $\wb_0,\psib, \wb_1, \xib$ satisfy 
\begin{equation}
  \label{eq.vtd_constraint_asymptotics}
  \partial_\theta \xib
      - 2 (a+\wb_1) (\del{\theta} \wb_0) \xib
      + \frac 12 (1+\psib)^{-1} (\del{\theta}\psib) \xib=0,
\end{equation}
this solution of the form \eqref{fuchsian_ansatz} is a solution to the full VTD system \eqref{vtd.u}--\eqref{vtd.H}. 
\end{prop}

\begin{proof}
  The system \eqref{vtd.1}--\eqref{vtd.4} implies the following system of equations for the new unknowns $\Omega$:
\begin{equation}
  \label{e.vtd_remainder}
  t\del{t}\Omega(t) - C \Omega(t) = \Phi(t) + \Hc(t,\Omega), 
\end{equation} 
where 
\begin{equation}
  C :=
  \begin{pmatrix} 
    0 & \frac{1}{a} & 0 & 0 \\
    0 & 0 & 0 & 0 \\
    0 & 0 & 0 & 0 \\
    0 & 0 & 0 & 0 
  \end{pmatrix},
\end{equation}
\begin{equation}
  \Phi(t) :=
  \begin{pmatrix} 
    0 \\ \phi_1(t) \\ \phi_2(t) \\ \phi_3(t)
  \end{pmatrix}
  = 
  \begin{pmatrix} 
    0 \\ 
    -\frac{m}{2}(a + \wb_1)(1+\psib)\xib^2 t^{2(\gamma-1)} \\ 
    -m(1+\psib)^2\xib^2 t^{2(\gamma-1)} \\ 
    -\frac{m}{4}(1+\psib)\xib^3 t^{3(\gamma-1)}
  \end{pmatrix},
\end{equation}
and in the limit $t\searrow 0$
\begin{equation}
  \label{Hdef}
  \Hc(t,\Omega(t, \theta)) :=
  \begin{pmatrix} 
    0 \\ 
    \hc_1 \\ 
    \hc_2 \\ 
    \hc_3 \\ 
  \end{pmatrix}
  =
  \begin{pmatrix} 
    0 \\ 
    O(t^{2(\gamma-1)}(\omega_1 +  \omega_2 + \omega_3)) + o(\Omega^2) \\ 
    O(t^{2(\gamma-1)}(\omega_2 + \omega_3)) + o(\Omega^2) \\ 
    O(t^{\gamma-1}(\omega_1 + \omega_3)+ t^{3(\gamma-1)}\omega_2) + o(\Omega^2) \\ 
  \end{pmatrix}.
\end{equation}
The $O$-notation (cf. footnote \ref{ft.oh_notation}) in \eqref{Hdef} is with respect to the $\norm{\cdot}_{H^k}$ norm.
For any $T_0 < 1$, $V,W \in  E_{k,\mu,T_0}$, and any $t \in (0,T_0)$, $\Hc$ satisfies $\|\Hc(t,V(t)) - \Hc(t,W(t))\|_{H^{k}} \le \|t^{\gamma-1}\|_{L^\infty} \|V(t) - W(t)\|_{H^{k}}$.
We note that $\|t^{\gamma-1}\|_{L^\infty} = t^{\gamma_{\mathrm{min}-1}}$, where $\gamma_{\mathrm{min}} := \min_{\theta\in \Tbb} \gamma(\theta)$.

Let $\mathfrak{F}[\Omega](t,\theta) := \Phi(t) + \Hc(t, \Omega(t,\theta))$, and formally define 
\begin{equation}
  \label{Gfdef}
  \Gf[V](t,\theta) := t^{C} \int_0^t s^{-1} s^{-C} \Ff[V](s,\theta)\rmd s.
\end{equation}
This quantity is the formal integral solution of \eqref{e.vtd_remainder}.
We show that $\Gf[\cdot]$ is a well-defined endomorphism on $E_{k, \mu, T}$, and is a contraction for $T_0$ sufficiently small.

Note that 
\begin{equation}
  t^{-C} :=
  \begin{pmatrix} 
    1 & -\frac 1a \ln(t) & 0 & 0 \\
    0 & 1 & 0 & 0 \\
    0 & 0 & 1 & 0 \\
    0 & 0 & 0 & 1 
  \end{pmatrix}.
\end{equation}
Since $\gamma>1, \mu > 0$, both $\del{s}(s^{-C}\Omega)$ and $s^{-1} s^{-C} \Ff[V](s,\theta)$ can be integrated on $[0, t]$, for any $t < T_0$.
Moreover, it follows from $\mu < 2(\gamma-1)$ that one can show $\Gf[0] \in E_{k,\mu,T_0}$. 
Thus $\Gf[\cdot]$ is well-defined, and the solution of the integral equation
\begin{equation}
  \label{vtd.weak_form}
  \Omega = \Gf[\Omega],
\end{equation}
should it exist, is differentiable in time and satisfies \eqref{vtd.1}--\eqref{vtd.4}.

We verify $\Gf[\cdot]$ exists by showing that it is a contraction mapping.
To show this, let $V$ and $W$ be any two elements of $E_{k,\mu,T_0}$.
Then,
\begin{align*}
  \sup_{t\in(0,T_0)} \|\Gf[V](t) - \Gf[W](t)\|_{H^{k}}
    &\le \sup_{t\in(0,T_0)}
      \left\{
        |t^C| \int_0^t s^{-1}|s^{-C}|\|\Ff[V](s)-\Ff[W](s)\|_{H^{k}} \rmd s
      \right\}, \\
    &\le \sup_{t\in(0,T_0)}
      \left\{
        |t^C| \int_0^t s^{\gamma_{\mathrm{min}}-1}s^{-1}|s^{-C}|\|V(s)-W(s)\|_{H^{k}} \rmd s
      \right\} \\
    &\le \sup_{t\in(0,T_0)}
      \left\{
        |t^C| \int_0^t s^{\gamma_{\mathrm{min}} + \mu -1}s^{-1}|s^{-C}|\|(V(s)-W(s))s^{-\mu}\|_{H^{k}} \rmd s
      \right\}, \\  
    &\le \sup_{t\in(0,T_0)}
      \left\{
        |t^C| \int_0^t s^{\gamma_{\mathrm{min}} + \mu -1}s^{-1}|s^{-C}|\rmd s
      \right\}
      \sup_{t\in(0,T_0)}\|(V(t)-W(t))t^{-\mu}\|_{H^{k}}.               
\end{align*}
In the last step, the factor $\|(V-W)t^{-\mu}\|_{H^{k}}$ is bounded on $[0,T_0]$ since $V,W \in  E_{k,\mu,T_0}$. 
The positivity of $\mu$ and $\gamma_{\min}-1$ implies that $\Zc(t) := |t^C| \int_0^t s^{\gamma_{\mathrm{min}} + \mu -1}s^{-1}|s^{-C}|\rmd s$ in the last equation above is bounded and vanishes as $t\searrow 0$. 
Taking $T_0$ sufficiently small it follows that $\sup_{t\in(0,T_0)} \Zc(t) < 1$, and thus that $\Gf[\cdot]$ is a contraction on $E_{k,\mu,T_0}$.
This shows existence of a unique solution to \eqref{vtd.weak_form} and by the arguments above, a time-differentiable solution to \eqref{e.vtd_remainder} in $E_{k,\mu,T_0}$, and finally the existence of solutions to the VTD system \eqref{vtd.1}-\eqref{vtd.4} of the form \eqref{fuchsian_ansatz}.

It remains to prove that the solution satisfies the full set of VTD equations \eqref{vtd.u}--\eqref{vtd.H}. 
Let $\Vc := \del{\theta} \nu - 2t \del{t}u \del{\theta}u + \frac 12 \del{\theta}\ln(\alpha)$.
A short calculation shows that $\del{t} \Vc = - \frac m2 \alpha t^{-3}e^{2\nu} \Vc$. 
Thus, the SIVP has a unique zero solution provided $\tilde \Vc = \lim_{t\searrow 0} \Vc = 0$. 
Analysis similar to that in Proposition~\ref{prop:hoexp} for the constraint in the full Einstein system shows that $\tilde \Vc = 0$ if and only if \eqref{eq.vtd_constraint_asymptotics} holds. 
This shows that the constraint is propagated by the VTD system under the assumption \eqref{eq.vtd_constraint_asymptotics} on $\wb_0,\psib, \wb_1, \xib$.
Solutions of the SIVP for the auxiliary equations \eqref{vtd.G}, \eqref{vtd.H} have been shown to exist in \cite{ames2013a}.
The results of that work apply here with $A\equiv 0$ (the metric field $A$ is denoted $Q$ in \eqref{T2metric} above).
\end{proof}

The following corollary is the analogue of Lemma~\ref{lem.time_derivatives} that holds for solutions of the VTD equations.
\begin{cor}
  \label{cor.vtd.time_derivatives}
  Assume that the conditions of Proposition~\ref{prop.vtd_existence} hold.
  Then the time derivatives $t\del{t}(\nu^{(VTD)})$, $t\del{t}H^{(VTD)}$, and $t\del{t}\psi^{(VTD)}$ satisfy
  \begin{align}
  \label{e.asymptotics.dnudt.vtd}
  \Bnorm{t\del{t}\nu^{(VTD)} - (a + \wb_1)^2}_{H^{k}} 
    &\lesssim t^\mu + t^{2(\gamma_{\mathrm{min}}-1-\sigma)}, \\
  \label{e.asymptotics.dHdt.vtd}
  \Bnorm{t\del{t}H^{(VTD)} - \sqrt{m}(1+\psib)^{1/2}t^{2(a+\wb_1)^2-2}\xib^2}_{H^{k}} &\lesssim t^{2(\gamma_{\mathrm{min}}-1-\sigma) + \mu}, \\
  \label{e.asymptotics.dadt.vtd}
  \Bnorm{t\del{t}\psi^{(VTD)} + \sqrt{m}(1+\psib)^2t^{2(a+\wb_1)^2-2}\xib^2}_{H^{k}} &\lesssim t^{2(\gamma_{\mathrm{min}}-1-\sigma)a + \mu},
  \end{align}  
  for any $\sigma >0$ which can be chosen arbitrarily small. 
\end{cor}
The proof of Corollary~\ref{cor.vtd.time_derivatives} follows essentially the same steps as the proof of Lemma~\ref{lem.time_derivatives}.

\subsection{Completion of the Proof of the Main Result}
\label{s.proof_main_result}

In this section we combine the results of the above subsections in order to complete the proof of Theorem~\ref{thm:mainresult}. 

\subsubsection{Existence, Decay, and Uniform Bounds}
\label{s.proof_main_result.part1}
Suppose $K$ satisfies $|K-1| > 2$ and $\sigma \in (0, 2\kappa_0/3)$, where $\kappa_0 = \min\{1, \frac 14(K-3)(K+1)\}$ are fixed. 
Let $T_0>0$ and $R_0>0$ be sufficiently small so that $|T_0^{\frac 14(K-3)(K+1)}| < R_0$.
Set $a = \frac 12(1-K)$. 
Clearly $|a| > 1$ is as required for Proposition~\ref{prop:main_existence} and \ref{prop:hoexp}.

Given any set of functions $(\mathring{u}, \mathring{\omega}, \mathring{\nu}, \mathring{\alpha}, \mathring{G}, \mathring{H})\in H^k(\Tbb, \Rbb^6)$, 
which serve as initial data for the polarized $\Tbb^2$-symmetric Einstein equations at $T_0$ in the sense that
$(u, \partial_t u, \nu, \alpha, G, H)|_{t=T_0} 
=(\mathring{u}, \mathring{\omega}, \mathring{\nu}, \mathring{\alpha}, \mathring{G}, \mathring{H})$, 
and which satisfy the Einstein constraint equations \eqref{Evac.nu.theta}, we define 
\begin{equation}
  \mathring{w_0} = \frac{1}{a}\mathring{u}, \quad
  \mathring{w_1} = T_0\mathring{\omega}-a, \quad
  \mathring{w_2} = T_0\del{\theta}\mathring{u}, \quad
  \mathring{\psi} = \mathring{\alpha} - 1, \quad
  \mathring{\eta} = T_0\del{\theta}\mathring{\alpha}, \quad
  \mathring{\xi} = T_0^{-1}e^{\mathring{\nu}}.
\end{equation}
Note that $\mathring{\xi} > 0$.
Assume that $\|\mathring{U}\|_{H^k} < R_0$ for $\mathring{U}=\bigl(\mathring{w}_1,\mathring{w}_2,\mathring{\psi},\mathring{\xi},\mathring{\eta}\bigr)^{\tr}$.
We may thus apply Proposition~\ref{prop:main_existence}, \ref{prop:hoexp}, and \ref{prop:fullEFE}.
This establishes the existence of a solution $u, \nu, \alpha, G, H$ to the full set of Einstein vacuum equations, and in addition the existence of functions $(\wt_0, \wt_1, \psit, \nut, \Ht)\in H^{k-1}(\Tbb), \AND \wt_2, \etat \in H^{k-2}(\Tbb)$ such that
\begin{equation}
  \label{eq.asymptotic_constraints}
  \wt_2 = \partial_\theta \wt_0, \quad
  \etat = \del{\theta}\psit, \AND
  \partial_\theta \nut
        - 2 (a+\wt_1) \wt_2
        + \frac 12 (1+\psit)^{-1} \tilde\eta 
        =0.
\end{equation}
Moreover, defining the \emph{asymptotic data}
\begin{equation*}
  \label{e.def.asymptotic_data}
  \ut = \wt_0, \quad
  \kt = 1- 2(\wt_1 + a), \AND
  \alphat = \psit,
\end{equation*}
the following decay estimates hold:
From Proposition~\ref{prop:main_existence}
\begin{align*}
  \Bnorm{t\del{t}u(t)- \frac 12 (1-\kt)}_{H^{k-1}}
    &\lesssim t+t^{2\kappa_0-2\sigma},\\
  \Bnorm{\alpha -1 - \alphat}_{H^{k-1}}
    &\lesssim t+t^{2\kappa_0-2\sigma},
\end{align*}
from Proposition~\ref{prop:hoexp}
\begin{align*}
  \norm{\del{\theta}u(t)-\frac 12 \ln(t)\del{\theta}\kt-\del{\theta}\ut}_{H^{k-2}} &\lesssim t+t^{2\kappa_0-3\sigma}, \\
  \norm{\nu(t)-\frac14(1-\kt)^2 \ln(t)-\nut}_{H^{k-1}}&\lesssim t+t^{2\kappa_0-2\sigma},\\  
  \norm{\del{\theta}\alpha-\del{\theta}\tilde{\alpha}}_{H^{k-2}} &\lesssim t+t^{2\kappa_0-3\sigma},
\end{align*}
and from Proposition~\ref{prop:fullEFE}
\begin{align*}
  \Bnorm{u(t)-\frac 12(1-\kt)\ln t-\ut}_{H^{k-1}}&\lesssim t+t^{2\kappa_0-2\sigma},\\
  \Bnorm{H(t)-\Ht}_{H^{k-1}}&\lesssim t^{2(\gamma_{\min}-1 -\sigma)},
\end{align*}
noting that in terms of $\kt$, $2(\gamma_{\min}-1 -\sigma) = \min_{\theta \in [0,2\pi)}\{\frac{1}{2}(\kt-3)(\kt+1)\} - 2\sigma$. 
In terms of these asymptotic data the constraints \eqref{eq.asymptotic_constraints} become 
\begin{align*}
  \wt_2 = \del{\theta}\ut, \quad 
  \etat = \del{\theta}\alphat, \AND 
  \del{\theta}\nut - (1-\kt)\del{\theta}\ut + \frac 12 \del{\theta}\ln(1+\alphat) = 0.
\end{align*}

In order to establish \eqref{eq:EstimateKasnerParameters}, we notice that $\norm{\kt-K}_{H^{k-1}}=\norm{2\wt_1}_{H^{k-1}}$,
using the definition of $\kt$ and the choice of $a$ above. Estimate \eqref{eq:EstimateKasnerParameters} then follows by taking the limit to $t=0$ of the estimate \eqref{eq:generalestimate} which one obtains as part of the proof of Proposition~\ref{prop:main_existence}.

Let $\mathcal{S}_{K,m,R_0}$ denote the set of solutions whose existence is established by the above application of Propositions~\ref{prop:main_existence}, \ref{prop:hoexp}, and \ref{prop:fullEFE}.
The discussion in Section~\ref{s.stability_of_kasner} shows that each solution in $\mathcal{S}_{K,m,R_0}$ can be considered a perturbation of the Kasner solution parameterized by $K$.
The estimates \eqref{thm.mainresult.kasnerstability.first}--\eqref{thm.mainresult.kasnerstability.last} are obtained by converting the estimates \eqref{s.stability_of_kasner.first}--\eqref{s.stability_of_kasner.last} back to the original metric variables.

\subsubsection{AVTD Property}

Fix a solution in $\mathcal{S}_{K,m,R_0}$ (as introduced above in Section~\ref{s.proof_main_result.part1}) given by $(u, \nu, \alpha, G, H)$ and having asymptotic data $\wt_0, \wt_1, \psit, \nut, \Ht$, and initial datum $\mathring{G}$.
Proposition \ref{prop.vtd_existence} yields the existence of a family of solutions of the form \eqref{polT2metric} to the velocity term dominated system \eqref{vtd.u}--\eqref{vtd.H} parameterized by functions $\wb_0, \wb_1, \psib, \xib, \Hb, \Gb$.
We show that the solution of the full Einstein vacuum equations asymptotically approaches the solution of the VTD system specified by the following choice of asymptotic data functions
\[
\wb_0 = \wt_0, \quad 
\wb_1=\wt_1, \quad 
\psib=\psit, \quad 
\xib=e^{\nut}, \quad 
\Hb=\Ht, \AND 
\Gb = \mathring{G}.
\]

As a result of this choice, and the results of Propositions~\ref{prop:main_existence}, \ref{prop:hoexp}, \ref{prop:fullEFE}, and Proposition~\ref{prop.vtd_existence}, as well as Lemma~\ref{lem.time_derivatives} and Corollary~\ref{cor.vtd.time_derivatives}, the right-hand sides of the following estimates are dominated by a positive power of $t$. 
For any $t \in (0, T_0]$,
\begin{align}
  \label{e.vtd_decay1}
  \begin{split}
  \Bnorm{(u, \alpha, \nu, H, G)- & (u^{(VTD)}, \alpha^{(VTD)}, \nu^{(VTD)}, H^{(VTD)}, G^{(VTD)})}_{H^{k-1}} \\
    \le & \Bnorm{u-(\wt_1 + a)\ln(t) + \wt_0}_{H^{k-1}} 
        + \Bnorm{u^{(VTD)}-(\wt_1 + a)\ln(t) + \wt_0}_{H^{k-1}} \\
        & + \Bnorm{\alpha-1-\alphat}_{H^{k-1}} 
        + \Bnorm{\alpha^{(VTD)}-1-\alphat}_{H^{k-1}} \\
        & +\Bnorm{\nu-(\wt_1 + a)^2\ln(t) - \nut}_{H^{k-1}} 
        + \Bnorm{\nu^{(VTD)}-(\wt_1 + a)^2\ln(t) - \nut}_{H^{k-1}} \\
        &  + \Bnorm{H-\Ht}_{H^{k-1}}
        + \Bnorm{H^{(VTD)}-\Ht}_{H^{k-1}},
        + \Bnorm{G-\Gt}_{H^{k-1}}
        + \Bnorm{G^{(VTD)}-\Gt}_{H^{k-1}},
  \end{split}
\end{align}
\begin{align}
  \label{e.vtd_decay3}
  \Bnorm{t\del{t}u-t\del{t}u^{(VTD)}}_{H^{k-1}}
    \le&  \Bnorm{t\del{t}u- \frac 12 (1-\kt)}_{H^{k-1}}
        + \Bnorm{t\del{t}u^{(VTD)} - \frac 12 (1-\kt)}_{H^{k-1}} \\     
  \label{e.vtd_decay4}
  \Bnorm{t\del{t}H - t\del{t}H^{(VTD)}}_{H^{k-1}}
    \le& \Bnorm{t\del{t}H - \sqrt{m}(1+\psit)^{1/2}t^{2(a+\wt_1)^2 -2}e^{2\nut}}_{H^{k-1}}\\
    &+ \Bnorm{t\del{t}H^{(VTD)} - \sqrt{m}(1+\psit)^{1/2}t^{2(a+\wt_1)^2 -2}e^{2\nut}}_{H^{k-1}} \\  
  \label{e.vtd_decay5} 
  \Bnorm{t\del{t}\nu - t\del{t}\nu^{(VTD)}}_{H^{k-1}}
    \le& \Bnorm{t\del{t}\nu - (a + \wt_1)^2 }_{H^{k-1}}
       + \Bnorm{t\del{t}\nu^{(VTD)} - (a + \wt_1)^2}_{H^{k-1}}, \\
  \label{e.vtd_decay6} 
      \Bnorm{t\del{t}\alpha - t\del{t}\alpha^{(VTD)}}_{H^{k-1}}
        \le& \Bnorm{t\del{t}\alpha + m(1 + \psit)^2t^{2(a + \wt_1)^2-2}e^{2\nut}}_{H^{k-1}}\\
        &+ \Bnorm{t\del{t}\alpha^{(VTD)} + m(1 + \psit)^2t^{2(a + \wt_1)^2-2}e^{2\nut}}_{H^{k-1}}    
\end{align}
Recall that $\del{t}G = 0$, and note that Proposition~\ref{prop.vtd_existence} and \eqref{nvars.1} imply that $\nu^{(VTD)} = (\wb_1 + a)^2\ln(t) + \ln(\xib) + \ln(1 + \omega_3/\xib)$.
This establishes the AVTD property in a weighted norm with $\beta = 1$ (cf. Definition~\ref{def.avtd}). 

\subsubsection{Curvature Blowup}
Finally, we show that each solution in $\mathcal{S}_{K,m,R_0}$ is inextendible as a $C^2$ metric past $t=0$. 
The Kretschmann scalar $\mathcal{K}$ can be straightforwardly computed with the help of computer algebra \cite{xact, xcoba}.
Evaluating $\mathcal{K}$ near $t \searrow 0$ using the expressions for the VTD expansion near the singularity, we find 
\begin{equation}
  \mathcal{K} \sim \frac{(\kt^2 -1)^2(3+\kt^2)e^{4(\ut-\nut)}}{4 \alphat}t^{-3-\kt^2}.  
\end{equation}
Recall definition \eqref{e.def.asymptotic_data} of $\kt, \ut, \vt, \alphat$. 
Clearly $\mathcal{K}$ is unbounded as $t \searrow 0$, which implies that the spacetime is inextendible as a $C^2$ manifold. 
The blow-up rate at each $\theta=\mathrm{const}$ hypersurface is the same as a Kasner spacetime with Kasner exponent given by $K = \kt(\theta)$. 
In fact, this calculation shows that any polarized $\Tbb^2$-symmetric spacetime which is AVTD is inextendible as a $C^2$ manifold in the contracting direction.

\appendix

\section{The Global Initial Value Problem for Symmetric Hyperbolic PDE Systems in Fuchsian Form}
\label{sec:Fuchsian}
As noted in the introduction, our results in this paper depend crucially on Theorem 3.8 from \cite{BOOS:2020}, which establishes the existence of solutions of the ``Global Initial Value Problem" (GIVP) for symmetric hyperbolic partial differential equations systems in Fuchsian form. 
In particular, this theorem shows that for a PDE system in the form
\begin{align}
  B^0(t,u)\del{t}u + B^i(t,u)\nabla_{i} u  &= \frac{1}{t}\Bc(t,u)\Pbb u + F(t,u) \quad \text{in $(0,T_0]\times \Sigma$,} \label{symivpA.1} 
\end{align}
with initial data
\begin{align}
  u &= u_0 \hspace{3.1cm} \text{in $\{T_0\}\times \Sigma$,} \label{symivpA.2}
\end{align}
specified at a time $T_0 >0$ on a closed manifold $\Sigma$, the solution $u(t)$ exist for all time from $T_0$ to the singular time $t=0$. Observe that in this general context $T_0$ is not necessarily required to be small.
Further discussion and applications of the GIVP for PDE's in Fuchsian form can be found in \cite{Oliynyk:CMP_2016,FOW:2020,LeFlochWei:2015,LiuOliynyk:2018b,LiuOliynyk:2018a,LiuWei:2019,Wei:2018,Oliynyk:2020}.
The particular form of the GIVP theorem which applies to our work in this paper is specified in Theorem~\ref{symthm}, which we present below. 
For general Fuchsian PDE systems, there is a rather long set of coefficient assumptions -- see \cite[\S 3.1.]{BOOS:2020} -- that need to be verified in order to apply this existence theorem. 
For the convenience of the reader, we state in Theorem~\ref{symthm} a simplified version of Theorem 3.8~from \cite{BOOS:2020} that is sufficient for the application considered in this article. 
The main point of this simplification is that it greatly reduces and simplifies the coefficient assumptions that need to be checked.
Observe that the convention in the presentation in \cite{BOOS:2020} is that the times $t$ and $T_0$ are negative. Since we work with the positive time convention in this paper exclusively we present all results in this appendix in terms of positive time intervals of the form $(0,T_0)$ or $[0,T_0]$. 

For our application here, the Fuchsian GIVP takes a simplified form as follows:
\begin{align}
B^0(t,u)\del{t}u + B^i(t,u)\del{i} u  &= \frac{1}{t}\Bc(t,u)\Pbb u + \frac{1}{t}F(t,u) \quad \text{in $(0,T_0]\times \Tbb^n$,} \label{symivpB.1} \\
u &= u_0 \hspace{3.4cm} \text{in $\{T_0\}\times \Tbb^n$,} \label{symivpB.2}
\end{align}
where $\del{i}=\del{}/\del{}x^i$, $i=1,2,\ldots, n$ are the partial derivatives with respect to the standard periodic coordinates $x=(x^i)$ on $\Tbb^n$ and the coefficients of the Fuchsian PDE system \eqref{symivpB.1} are assumed to satisfy the following hypotheses.
We state these for PDE systems of general dimension, but note that in our application above $N = 5$ and $n=1$.

\subsection{Assumptions on the Coefficients of the PDE System}
 \label{coeffassumps}
\begin{enumerate}[(i)]
\item $T_0>0$ and $u(t,x)$ is an $\Rbb^N$-valued map.
\item $\Pbb\in \Mbb{N}$ is\footnote{We denote by $\Mbb{N}$ the collection of constant $N\times N$ matrices with real entries.} a symmetric projection
operator; that is,
\begin{equation*} 
\Pbb^2 = \Pbb,  \quad  \Pbb^{\tr} = \Pbb, \quad \del{t}\Pbb =0. 
\end{equation*}
For use below, we define the \textit{complementary projection operator}; that is
\begin{equation*}
\Pbb^\perp = \id -\Pbb,
\end{equation*}
which by our above assumptions, is also a constant, symmetric projection operator.

\item There exist constants $R, \kappa, \gamma_1, \gamma_2 >0$ such that the matrix valued maps 
\begin{equation*}
B^0 \in 
C^0\bigl([0,T_0], C^\infty(B_R(\Rbb^N),\Mbb{N})\bigr)\cap C^1\bigl((0,T_0], C^\infty(B_R(\Rbb^N),\Mbb{N})\bigr)
\end{equation*}
and $\Bc\in C^0\bigl([0,T_0], C^\infty(B_R(\Rbb^N),\Mbb{N})\bigr)$  satisfy
\begin{gather*}
\frac{1}{\gamma_1} \id \leq  B^0(t,v)\leq \frac{1}{\kappa} \Bc(t,v) \leq  \gamma_2\id, \\
[\Pbb,\Bc(t,v)] = 0,\\
B^0(t,v)^{\tr} = B^0(t,v)
\intertext{and}
\Pbb^\perp B^0(t,v)\Pbb = \Pbb B^0(t,v) \Pbb^\perp=0
\end{gather*}
for  all $(t,v)\in (0,T_0]\times B_{R}(\Rbb^N)$.
\bigskip

\item The vector valued map $F\in C^0\bigl([0,T_0], C^\infty(B_R(\Rbb^N),\Rbb^N)\bigr)$ satisfies
\begin{equation} 
\Pbb F(t,v) = 0 
\end{equation}
and for all $(t,v)\in [0,T_0]\times B_R(\Rbb^N)$ there exists a constant $\lambda\geq 0$
such that
\begin{align*}
\Pbb^\perp F(t,v) = \Ordc\biggl(\frac{\lambda}{R}|\Pbb v|^2 \biggr)
\end{align*}
for all $(t,v)\in  (0,T_0]\times B_R(\Rbb^N)$.

\bigskip

\item The matrix valued maps $B^k\in C^0\bigl([0,T_0], C^\infty(B_R(\Rbb^N),\Mbb{N})\bigr)$, $k=1,2,\ldots n$, satisfy
\begin{equation*}
B(t,v)^{\tr}=B(t,v)
\end{equation*}
for all $(t,v)\in (0,T_0]\times B_R(\Rbb^N)$. \

\bigskip

\item  There exist constants $\theta,\beta\geq 0$ such that the  map
\begin{equation*}
\Div\! B \: : \: (0,T_0]\times B_R(\Rbb^N) \times B_R(\Rbb^{N\times n})  \longrightarrow \Mbb{N}
\end{equation*} 
defined by ($I,J,\ldots=1,\ldots,N$; $i,j,\ldots=1,\ldots,n$; the Einstein summation convention is assumed)
\begin{equation}
  \label{DivB}
  \begin{split}
  (\Div\!  B)^I_J&(t,v,w) := (\del{t}{B^0})^I_J(t,v) + (\del{v^K} B^i)^I_J(t,v) w_i^K\\
  &+ (\del{v^K} B^0)^I_J(t,v) ((B^0)^{-1})^K_L(t,v)\Bigl[-(B^j)^L_M(t,v) w_j^M 
+
\frac{1}{t}\Bc^L_M(t,v)\Pbb^M_S v^S + \frac{1}{t}F^L(t,v) \Bigr]
\end{split}
\end{equation}
satisfies
\begin{align*}
 \Div\! B(t,v,w) &= \Ordc\Bigl(\theta+\frac{\beta}{t}|\Pbb v|^2 \Bigr)
\end{align*}
for all $(t,v,w)\in (0,T_0]\times B_R(\Rbb^N)\times B_R(\Rbb^{N\times n})$, where $v=(v^I)$, $w=(w^I_i)$.

\bigskip

\noindent \textit{Note:} It is not difficult to verify that 
\begin{equation*}
\Div\! B(t,u(t,x),D_x u(t,x)) =\del{t}(B^0(t,u(t,x))+ \del{i}(B^i(t,u(t,x)))
\end{equation*}
for solutions $u(t,x)$ of \eqref{symivpB.1}. 
\end{enumerate}

\subsection{Existence and Uniqueness Theorem for the Global Initial Value Problem for the Symmetric Hyperbolic Fuchsian PDE System}

We are now ready to state the existence theorem for the Fuchsian GIVP that we employ in this article. It is a special case of Theorem 3.8.~from \cite{BOOS:2020}, where we are using the improvement to the decay estimate as discussed in Remark 3.10.(a).2~from \cite{BOOS:2020}. One important difference to notice in the theorem below compared to that of  \cite[Theorem 3.8.]{BOOS:2020} is that the regularity requirement is lower (i.e., $k \in \Zbb_{>n/2+1}$ versus $k\in \Zbb_{>n/2+2}$). This is due to the matrix terms $B^i$ being regular in $t$ (as specified in condition (v) above) for the Fuchsian systems that we are considering, and it is this feature that allows the application of \cite[Lemma 3.5.]{BOOS:2020} to be avoided in the existence proof, which leads to the reduction in the required regularity.

\begin{thm} \label{symthm}
Suppose $k \in \Zbb_{>n/2+1}$, $\sigma>0$, $u_0\in H^k(\Tbb^n,\Rbb^N)$, assumptions (i)-(vi) from Section \ref{coeffassumps} are fulfilled,
and the constants $\kappa$, $\gamma_1$, and $\lambda$ from Section \ref{coeffassumps} satisfy $\kappa > \gamma_1(\lambda+\beta/2)$.
Then there exists
a $\delta > 0$ such that if $\norm{u_0}_{H^k(\Tbb^n)} \leq \delta$,
then there exists a unique solution 
\begin{equation*}
u \in C^0\bigl((0,T_0],H^k(\Tbb^n,\Rbb^N)\bigr)\cap L^\infty\bigl((0,T_0],H^k(\Tbb^n,\Rbb^N)\bigr)\cap C^1\bigl((0,T_0],H^{k-1}(\Tbb^n,\Rbb^N)\bigr)
\end{equation*}
of the GIVP \eqref{symivpB.1}-\eqref{symivpB.2} such that the limit $\lim_{t\searrow 0} \Pbb^\perp u(t)$, denoted $\Pbb^\perp u(0)$, exists in $H^{k-1}(\Tbb^n,\Rbb^N)$.

\medskip

\noindent Moreover, for $0<t\le T_0$,  the solution $u$ satisfies  the energy estimate
\begin{equation}
  \label{eq:generalestimate}
\norm{u(t)}_{H^k(\Tbb^n)}^2+ \int^{T_0}_t \frac{1}{\tau} \norm{\Pbb u(\tau)}_{H^k(\Tbb^n)}^2\, d\tau   \lesssim \norm{u_0}_{H^k(\Tbb^n)}^2
\end{equation}
and the decay estimates
\begin{align*}
\norm{\Pbb u(t)}_{H^{k-1}(\Tbb^n)} &\lesssim \begin{cases}
t & \text{if $\kappa > 1$} \\
t^{\kappa-\sigma} & \text{if $0 < \kappa \leq 1$}
\end{cases}
\intertext{and}
\norm{\Pbb^\perp u(t) - \Pbb^\perp u(0)}_{H^{k-1}(\Tbb^n)} &\lesssim
 \begin{cases}  t
& \text{if $\kappa > 1$} \\
 t+t^{2(\kappa-\sigma)}  &  \text{if $0<\kappa \leq 1 $ }
 \end{cases}.
\end{align*}
\end{thm}

It is important to be aware that the constant $\delta$ as well as the implicit constants in the estimates in this theorem, in general, depend implicitly on the choices of $k$, $\sigma$, and all the quantities introduced in  assumptions (i)-(vi) from Section \ref{coeffassumps}, in particular, $T_0$, $R$, and $\kappa$. The proof of Theorem~\ref{symthm}, allowing for the change in regularity noted above, follows essentially the same steps as the proof of Theorem 3.8 in reference \cite{BOOS:2020}.

\bibliographystyle{habbrv}
\bibliography{mainbib}

\begin{thebibliography}{10}
\expandafter\ifx\csname url\endcsname\relax
  \def\url#1{\texttt{#1}}\fi
\expandafter\ifx\csname doi\endcsname\relax
  \def\doi#1{\burlalt{doi:#1}{http://dx.doi.org/#1}}\fi
\expandafter\ifx\csname arXiv\endcsname\relax
  \def\arxiv#1{arXiv:\burlalt{#1}{http://arxiv.org/abs/#1}}\fi
\expandafter\ifx\csname urlprefix\endcsname\relax\def\urlprefix{URL }\fi
\expandafter\ifx\csname href\endcsname\relax
  \def\href#1#2{#2}\fi
\expandafter\ifx\csname burlalt\endcsname\relax
  \def\burlalt#1#2{\href{#2}{#1}}\fi

\bibitem{ames2013a}
E.~Ames, F.~Beyer, J.~Isenberg, and P.~G. LeFloch.
\newblock Quasilinear {{Hyperbolic Fuchsian Systems}} and {{AVTD Behavior}} in
  {$T^2$}-{{Symmetric Vacuum Spacetimes}}.
\newblock {\em Ann. Henri Poincar{\'e}}, 14(6):1445--1523, 2013.
\newblock \doi{10.1007/s00023-012-0228-2}.

\bibitem{ABIO2021}
E.~Ames, F.~Beyer, J.~Isenberg, and T.~Oliynyk.
\newblock {Nonlinear stability of polarised $T^2$-symmetric spacetimes with a
  cosmological constant}.
\newblock In preparation, 2021.

\bibitem{andersson2001}
L.~Andersson and A.~D. Rendall.
\newblock Quiescent {{Cosmological Singularities}}.
\newblock {\em Commun. Math. Phys.}, 218(3):479--511, 2001.
\newblock \doi{10.1007/s002200100406}.

\bibitem{Andersson:2005}
L.~Andersson, H.~van Elst, W.~C. Lim, and C.~Uggla.
\newblock Asymptotic silence of generic cosmological singularities.
\newblock {\em Phys.\ Rev.\ Lett.}, 94(5):051101, 2005.
\newblock \doi{10.1103/PhysRevLett.94.051101}.

\bibitem{belinskii1970}
V.~A. Belinskii, I.~M. Khalatnikov, and E.~M. Lifshitz.
\newblock Oscillatory approach to a singular point in the relativistic
  cosmology.
\newblock {\em Adv. Phys.}, 19(80):525--573, 1970.
\newblock \doi{10.1080/00018737000101171}.

\bibitem{BERGER1997}
B.~K. Berger, P.~T. Chru\'{s}ciel, J.~Isenberg, and V.~Moncrief.
\newblock Global foliations of vacuum spacetimes with {$T^2$} isometry.
\newblock {\em Ann.\ Phys.}, 260(1):117--148, 1997.
\newblock \doi{10.1006/aphy.1997.5707}.

\bibitem{Berger:2019}
B.~K. Berger, J.~Isenberg, and A.~Layne.
\newblock Stability within {$T^2$}-symmetric expanding spacetimes.
\newblock {\em Ann. Henri Poincar\'{e}}, 21(3):675--703, 2020.
\newblock \doi{10.1007/s00023-019-00870-8}.

\bibitem{beyer2017}
F.~Beyer and P.~G. LeFloch.
\newblock Self\textendash{}gravitating fluid flows with {{Gowdy}} symmetry near
  cosmological singularities.
\newblock {\em Commun. Partial. Differ. Equ.}, 42(8):1199--1248, 2017.
\newblock \doi{10.1080/03605302.2017.1345938}.

\bibitem{BOOS:2020}
F.~Beyer, T.~A. Oliynyk, and J.~A. Olvera-Santamaría.
\newblock The fuchsian approach to global existence for hyperbolic equations.
\newblock {\em Commun. Partial. Differ. Equ.}, 0(0):1--82, 2020.
\newblock \doi{10.1080/03605302.2020.1857402}.

\bibitem{choquet-bruhat2006}
Y.~{Choquet-Bruhat} and J.~Isenberg.
\newblock Half polarized {$U(1)$}-symmetric vacuum spacetimes with {{AVTD}}
  behavior.
\newblock {\em J. Geom. Phys.}, 56(8):1199--1214, 2006.
\newblock \doi{10.1016/j.geomphys.2005.06.011}.

\bibitem{choquet-bruhat2004}
Y.~{Choquet-Bruhat}, J.~Isenberg, and V.~Moncrief.
\newblock Topologically general {$U(1)$} symmetric vacuum space-times with
  {{AVTD}} behavior.
\newblock {\em Nuovo Cim. B}, 119(7-9):625--638, 2004.
\newblock \doi{10.1393/ncb/i2004-10174-x}.

\bibitem{CIM1990}
P.~T. Chru{\'s}ciel, J.~Isenberg, and V.~Moncrief.
\newblock Strong cosmic censorship in polarised gowdy spacetimes.
\newblock {\em Class. Quantum Grav.}, 7(10):1671--1680, 1990.
\newblock \doi{10.1088/0264-9381/7/10/003}.

\bibitem{ChruscielKlinger:2015}
P.~T. Chr\'{u}sciel and P.~Klinger.
\newblock Vacuum spacetimes with controlled singularities and without
  symmetries.
\newblock {\em Phys. Rev. D}, 92:041501, 2015.
\newblock \doi{10.1103/PhysRevD.92.041501}.

\bibitem{CHRUSCIEL:1990}
P.~T. Chruściel.
\newblock On space-times with {$U(1)\times U(1)$} symmetric compact {Cauchy}
  surfaces.
\newblock {\em Ann.\ Phys.}, 202(1):100 -- 150, 1990.
\newblock \doi{10.1016/0003-4916(90)90341-K}.

\bibitem{ClausenThesis2007}
A.~Clausen.
\newblock {\em Singular Behavior in $T^2$ Symmetric Spacetimes with
  Cosmological Constant}.
\newblock PhD thesis, University of Oregon, 2007.

\bibitem{Clausen2007}
A.~Clausen and J.~Isenberg.
\newblock Areal foliation and asymptotically velocity-term dominated behavior
  in t2 symmetric space-times with positive cosmological constant.
\newblock {\em J. Math. Phys.}, 48(8):082501, 2007.
\newblock \doi{10.1063/1.2767534}.

\bibitem{xcoba}
{D. Yllanes and J. M. Mart\'{i}n-Garc\'{i}a}.
\newblock {xCoba}: General component tensor computer algebra.
\newblock \urlprefix\url{http://www.xact.es/xCoba}.

\bibitem{damour2002}
T.~Damour, M.~Henneaux, A.~D. Rendall, and M.~Weaver.
\newblock Kasner-{{Like Behaviour}} for {{Subcritical Einstein}}-{{Matter
  Systems}}.
\newblock {\em Ann. Henri Poincar{\'e}}, 3(6):1049--1111, 2002.
\newblock \doi{10.1007/s000230200000}.

\bibitem{Eardley:1972}
D.~M. Eardley, E.~Liang, and R.~K. Sachs.
\newblock Velocity-dominated singularities in irrotational dust cosmologies.
\newblock {\em J. Math. Phys.}, 13(1):99, 1972.
\newblock \doi{10.1063/1.1665859}.

\bibitem{FOW:2020}
D.~Fajman, T.~Oliynyk, and Z.~Wyatt.
\newblock {Stabilizing Relativistic Fluids on Spacetimes with Non-Accelerated
  Expansion}.
\newblock {\em Commun. Math. Phys.}, 2021.
\newblock \doi{10.1007/s00220-020-03924-9}.

\bibitem{Fournodavlos:2020}
G.~Fournodavlos and J.~Luk.
\newblock Asymptotically kasner-like singularities.
\newblock 2020, \arxiv{2003.13591v1}.

\bibitem{fournodavlos2020b}
G.~Fournodavlos, I.~Rodnianski, and J.~Speck.
\newblock Stable {Big} {Bang} formation for {Einstein}'s equations: {The}
  complete sub-critical regime.
\newblock 2020, \arxiv{2012.05888}.

\bibitem{Geroch:1972jmp}
R.~Geroch.
\newblock A method for generating new solutions of einstein's equation. {II}.
\newblock {\em J. Math. Phys.}, 13(3):394--404, 1972.
\newblock \doi{10.1063/1.1665990}.

\bibitem{Gowdy1974}
R.~H. Gowdy.
\newblock Vacuum spacetimes with two-parameter spacelike isometry groups and
  compact invariant hypersurfaces: Topologies and boundary conditions.
\newblock {\em Ann. Phys.}, 83(1):203--241, 1974.
\newblock \doi{10.1016/0003-4916(74)90384-4}.

\bibitem{heinzle2012}
J.~M. Heinzle and P.~Sandin.
\newblock The {{Initial Singularity}} of {{Ultrastiff Perfect Fluid Spacetimes
  Without Symmetries}}.
\newblock {\em Commun. Math. Phys.}, 313(2):385--403, 2012.
\newblock \doi{10.1007/s00220-012-1496-x}.

\bibitem{heinzle2012a}
J.~M. Heinzle, C.~Uggla, and W.~C. Lim.
\newblock Spike oscillations.
\newblock {\em Phys. Rev. D}, 86(10):104049, 2012.
\newblock \doi{10.1103/PhysRevD.86.104049}.

\bibitem{isenberg1999}
J.~Isenberg and S.~Kichenassamy.
\newblock Asymptotic behavior in polarized {$T^2$}-symmetric vacuum
  space\textendash{}times.
\newblock {\em J. Math. Phys.}, 40(1):340--352, 1999.
\newblock \doi{10.1063/1.532775}.

\bibitem{Isenberg:1990}
J.~Isenberg and V.~Moncrief.
\newblock Asymptotic behavior of the gravitational field and the nature of
  singularities in gowdy spacetimes.
\newblock {\em Ann.\ Phys.}, 199(1):84--122, 1990.
\newblock \doi{10.1016/0003-4916(90)90369-Y}.

\bibitem{isenberg2002}
J.~Isenberg and V.~Moncrief.
\newblock Asymptotic behaviour in polarized and half-polarized {$U(1)$}
  symmetric vacuum spacetimes.
\newblock {\em Class. Quantum Grav.}, 19(21):5361--5386, 2002.
\newblock \doi{10.1088/0264-9381/19/21/305}.

\bibitem{xact}
{J. M. Mart\'{i}n-Garc\'{i}a}.
\newblock {xAct}: Efficient tensor computer algebra for the {Wolfram} language.
\newblock \urlprefix\url{http://www.xact.es}.

\bibitem{kasner1921}
E.~Kasner.
\newblock Geometrical {{Theorems}} on {{Einstein}}'s {{Cosmological
  Equations}}.
\newblock {\em Am. J. Math.}, 43(4):217, 1921.
\newblock \doi{10.2307/2370192}.

\bibitem{kichenassamy2007k}
S.~Kichenassamy.
\newblock {\em Fuchsian {{Reduction}}}, volume~71 of {\em Progress in
  {{Nonlinear Differential Equations}} and {{Their Applications}}}.
\newblock Birkh{\"a}user Boston, Boston, MA, 2007.
\newblock \doi{10.1007/978-0-8176-4637-0}.

\bibitem{kichenassamy1998}
S.~Kichenassamy and A.~D. Rendall.
\newblock Analytic description of singularities in {{Gowdy}} spacetimes.
\newblock {\em Class. Quantum Grav.}, 15(5):1339--1355, 1998.
\newblock \doi{10.1088/0264-9381/15/5/016}.

\bibitem{LeFloch:2016}
P.~G. LeFloch and J.~Smulevici.
\newblock Future asymptotics and geodesic completeness of polarized
  t2-symmetric spacetimes.
\newblock {\em Ann.\ P.D.E.}, 9(2):363--395, 2016.
\newblock \doi{10.2140/apde.2016.9.363}.

\bibitem{LeFlochWei:2015}
P.~G. LeFloch and C.~Wei.
\newblock The global nonlinear stability of self-gravitating irrotational
  {C}haplygin fluids in a {FRW} geometry.
\newblock 2015, \arxiv{1512.03754}.

\bibitem{lifshitz1963}
E.~M. Lifshitz and I.~M. Khalatnikov.
\newblock Investigations in relativistic cosmology.
\newblock {\em Adv. Phys.}, 12(46):185--249, 1963.
\newblock \doi{10.1080/00018736300101283}.

\bibitem{LiuOliynyk:2018b}
C.~Liu and T.~A. Oliynyk.
\newblock Cosmological {N}ewtonian limits on large spacetime scales.
\newblock {\em Commun. Math. Phys.}, 364:1195--1304, 2018.
\newblock \doi{10.1007/s00220-018-3214-9}.

\bibitem{LiuOliynyk:2018a}
C.~Liu and T.~A. Oliynyk.
\newblock Newtonian limits of isolated cosmological systems on long time
  scales.
\newblock {\em Ann. \ Henri \ Poincar{\'e}}, 19:2157--2243, 2018.
\newblock \doi{10.1007/s00023-018-0686-2}.

\bibitem{LiuWei:2019}
C.~Liu and C.~Wei.
\newblock Future stability of the {FLRW} spacetime for a large class of perfect
  fluids.
\newblock 2019, \arxiv{1810.11788}.

\bibitem{Lott:2020b}
J.~Lott.
\newblock Kasner-like regions near crushing singularities.
\newblock 2020, \arxiv{2008.02674v2}.

\bibitem{Lott:2020a}
J.~Lott.
\newblock On the initial geometry of a vacuum cosmological spacetime.
\newblock {\em Class. Quantum Grav.}, 37(8):085017, 2020.
\newblock \doi{10.1088/1361-6382/ab77eb}.

\bibitem{Oliynyk:CMP_2016}
T.~A. Oliynyk.
\newblock Future stability of the {FLRW} fluid solutions in the presence of a
  positive cosmological constant.
\newblock {\em Commun. Math. Phys.}, 346:293--312, 2016.
\newblock \doi{10.1007/s00220-015-2551-1}.
\newblock See the preprint [arXiv:1505.00857] for a corrected version.

\bibitem{Oliynyk:2020}
T.~A. Oliynyk.
\newblock Future global stability for relativistic perfect fluids with linear
  equations of state $p={K}\rho$ where $1/3<{K}<1/2$.
\newblock 2020, \arxiv{2002.12526}.

\bibitem{ringstrom2009a}
H.~Ringstr{\"o}m.
\newblock Strong cosmic censorship in ${T}^3$-{Gowdy} spacetimes.
\newblock {\em Ann. Math.}, 170(3):1181--1240, 2009.
\newblock \doi{10.4007/annals.2009.170.1181}.

\bibitem{Ringstrom:2015}
H.~Ringstr{\"o}m.
\newblock Instability of spatially homogeneous solutions in the class of
  $\mathbb{T}^{2}$-symmetric solutions to {Einstein's} vacuum equations.
\newblock {\em Commun. Math. Phys.}, 334(3):1299--1375, 2015.
\newblock \doi{10.1007/s00220-014-2258-8}.

\bibitem{ringstrom2017}
H.~Ringstr\"{o}m.
\newblock Linear systems of wave equations on cosmological backgrounds with
  convergent asymptotics.
\newblock {\em Ast\'{e}risque}, (420):1--526, 2020.
\newblock \doi{10.24033/ast.1123}.

\bibitem{Ringstrom:2021a}
H.~Ringstr\"{o}m.
\newblock On the geometry of silent and anisotropic big bang singularities,
  2021, \arxiv{2101.04955v1}.

\bibitem{Ringstrom:2021b}
H.~Ringstr{\"o}m.
\newblock Wave equations on silent big bang backgrounds, 2021,
  \arxiv{2101.04939v1}.

\bibitem{Rodnianski2018HighD}
I.~{Rodnianski} and J.~{Speck}.
\newblock {On the nature of Hawking's incompleteness for the Einstein-vacuum
  equations: The regime of moderately spatially anisotropic initial data}.
\newblock 2018, \arxiv{1804.06825}.

\bibitem{rodnianski2018}
I.~Rodnianski and J.~Speck.
\newblock A regime of linear stability for the {{Einstein}}-scalar field system
  with applications to nonlinear {{Big Bang}} formation.
\newblock {\em Ann. Math.}, 187(1):65--156, 2018.
\newblock \doi{10.4007/annals.2018.187.1.2}.

\bibitem{rodnianski2014}
I.~Rodnianski and J.~Speck.
\newblock Stable {{Big Bang}} formation in near-{{FLRW}} solutions to the
  {{Einstein}}-scalar field and {{Einstein}}-stiff fluid systems.
\newblock {\em Sel. Math. New Ser.}, 24(5):4293--4459, 2018.
\newblock \doi{10.1007/s00029-018-0437-8}.

\bibitem{stahl2002}
F.~St{\aa}hl.
\newblock Fuchsian analysis of {$S^2\times S^1$} and {$S^3$} {{Gowdy}}
  spacetimes.
\newblock {\em Class. Quantum Grav.}, 19(17):4483--4504, 2002.
\newblock \doi{10.1088/0264-9381/19/17/301}.

\bibitem{TaylorIII:1996}
M.~E. Taylor.
\newblock {\em Partial differential equations {III}: {N}onlinear equations}.
\newblock Springer, 1996.
\newblock \doi{10.1007/978-1-4419-7049-7}.

\bibitem{uggla2003}
C.~Uggla, H.~van Elst, J.~Wainwright, and G.~F.~R. Ellis.
\newblock Past attractor in inhomogeneous cosmology.
\newblock {\em Phys. Rev. D}, 68(10):938, 2003.
\newblock \doi{10.1103/PhysRevD.68.103502}.

\bibitem{wainwright1997}
J.~Wainwright and G.~F.~R. Ellis, editors.
\newblock {\em Dynamical {Systems} in {Cosmology}}.
\newblock Cambridge University Press, New York, 1997.
\newblock
  \urlprefix\url{http://www.cambridge.org/gb/knowledge/isbn/item1152387/?site_locale=en_GB}.

\bibitem{Weaver:2001}
M.~{Weaver}, B.~K. {Berger}, and J.~{Isenberg}.
\newblock {Oscillatory Approach to the Singularity in Vacuum T$^{2}$ Symmetric
  Spacetimes}.
\newblock In V.~G. {Gurzadyan}, R.~T. {Jantzen}, and R.~{Ruffini}, editors,
  {\em The Ninth Marcel Grossmann Meeting}, pages 1011--1012, 2002.
\newblock \doi{10.1142/9789812777386_0140}.

\bibitem{Wei:2018}
C.~Wei.
\newblock Stabilizing effect of the power law inflation on isentropic
  relativistic fluids.
\newblock {\em J.\ Differ.\ Equ.}, 265:3441 -- 3463, 2018.
\newblock \doi{10.1016/j.jde.2018.05.007}.

\end{thebibliography}

\end{document}